\documentclass[a4paper,UKenglish,cleveref, autoref, thm-restate]{lipics-v2021}

\usepackage{thmtools} 
\usepackage{colortbl}
\usepackage[table]{xcolor}
\usepackage{thm-restate}
\usepackage{hyperref}
\usepackage[T1]{fontenc}
\usepackage{microtype}
\usepackage{graphicx} 
\usepackage{tabularx}

\usepackage{amssymb}
\usepackage{amsfonts}
\usepackage{thmtools, thm-restate}
\usepackage{mathtools}
\usepackage{tikz} 
\usetikzlibrary{arrows.meta} 
\usetikzlibrary{decorations.pathreplacing} 
\usepackage{cleveref}
\usepackage{xcolor}
\usepackage{booktabs} 
\usepackage[ruled,vlined]{algorithm2e}
\usepackage{multirow}
\usepackage{makecell}

\SetKw{KwAnd}{and}
\SetKw{KwAnd}{or}
\SetKw{KwAnd}{not}
\DontPrintSemicolon
\renewcommand{\thealgocf}{}
\SetKwFor{For}{for}{\!:}{endfor}
\SetKwFor{While}{while}{\!:}{endw}
\SetKwFor{ForEach}{for each}{\!:}{endfor}
\SetKwIF{If}{ElseIf}{Else}{if}{\!:}{else if}{else}{endif}
\SetKwComment{tcp}{(}{)}
\SetKwInOut{Input}{input}
\SetKwInOut{Output}{output}

\newcommand{\R}{\mathbb{R}}
\newcommand{\N}{\mathbb{N}}
\newcommand{\OPT}{\textrm{OPT}}
\newcommand{\ffunc}{f} 
\newcommand{\gfunc}{g} 
\newcommand{\hfunc}{h} 
\newcommand{\setF}{F} 
\newcommand{\setG}{G} 
\newcommand{\setH}{H} 
\newcommand{\comp}{\mathrm{c}} 
\newcommand{\SR}{\tilde{\gamma}} 
\newcommand{\e}{\mathrm{e}}
\DeclareMathOperator*{\argmax}{arg\,max}

\DeclarePairedDelimiter{\abs}{\lvert}{\rvert}

\hideLIPIcs
\nolinenumbers

\title{Incremental--Decremental Maximization}
\author{Yann Disser}{TU Darmstadt, Germany}{disser@mathematik.tu-darmstadt.de}{https://orcid.org/0000-0002-2085-0454}{}
\author{Max Klimm}{TU Berlin, Germany}{klimm@math.tu-berlin.de}{https://orcid.org/0000-0002-9061-2267}{}
\author{Annette Lutz}{TU Darmstadt, Germany}{lutz@mathematik.tu-darmstadt.de}{https://orcid.org/0009-0008-7699-7018}{}
\author{Lea Strubberg}{TU Berlin, Germany}{strubberg@math.tu-berlin.de}{https://orcid.org/0009-0009-8505-3614}{}
\authorrunning{Y.~Disser, M.~Klimm, A.~Lutz, L.~Strubberg}
\Copyright{Yann~Disser, Max~Klimm, Annette~Lutz and Lea~Strubberg} 
\funding{Supported by Deutsche Forschungsgemeinschaft (DFG, German Research Foundation) through subprojects A07 and A09 of CRC/TRR154.}
\acknowledgements{
    We are grateful to Martin Knaack, Jannik Matuschke, and David Weckbecker for helpful discussions on functions of bounded curvature, gross substitute functions and incremental maximization, respectively. We also thank Júlia Baligács and Nils Mosis for an initial exchange about the topic.}

\begin{document}
\begin{titlepage}
    \maketitle

    \begin{abstract}
        We introduce a framework for incremental--de\-cre\-men\-tal maximization that captures the gradual transformation or renewal of infra\-structures.
        In our model, an initial solution is transformed one element at a time and the utility of an intermediate solution is given by the sum of the utilities of the transformed and untransformed parts.
        We propose a simple randomized and a deterministic algorithm that both find an order in which to transform the elements while maintaining  a large utility during all stages of transformation, relative to an optimum solution for the current stage.
        More specifically, our algorithms yield competitive solutions for utility functions of bounded curvature and/or generic submodularity ratio, and, in particular, for submodular functions, and gross substitute functions.
        Our results exhibit that incremental--decremental maximization is substantially more difficult than incremental maximization.
    \end{abstract}
\end{titlepage}

\setcounter{page}{1}

\section{Introduction}

Incremental maximization provides a mathematical model to study situations where a new infrastructure is constructed over time.
The infrastructure consists of a set of elements~$E$.
The benefit for having built a subset of elements is defined by a set function~$f \colon 2^E \to \R_{\geq 0}$ that maps each partial solution $S \subseteq E$ to its benefit $f(S)$.
Depending on the application, the elements may model, e.g.,  pipes of a pipe system, wires in an electric network, or machines in a workshop. The benefit $f(S)$ models, e.g., averaged annual returns from operating the partial infrastructure $S$.
Given such a situation, it is a natural question in which order the elements should be constructed.
This question is addressed with the theory of incremental maximization problems. 

Formally, an incremental solution is an order $\pi = (e_1,\dots,e_n)$ of the elements of $E$ where~$n = |E|$. The quality of an incremental solution is measured by a competitive analysis. To this end, for $k \in \{1,\dots,n\}$, let 
\[\OPT_k \coloneqq \max \bigl\{ f(S) : S \subseteq E \text{ with } |S|=k \bigr\}\]
be the value of an optimal solution with $k$ elements. 
For an incremental solution~$\pi = (e_1,\dots,e_n)$ and $k \in \{1,\dots,n\}$, let $S_k \coloneqq \{e_1,\dots,e_k\}$ denote the prefix of the first~$k$ elements of $\pi$.
Then, the competitive ratio of an incremental solution~$\pi$ is defined as $\rho \coloneqq \inf \bigl\{ r \in \R_{\geq 0} : r\,f(S_k) \geq \OPT_k \text{ for all } k \in \{1,\dots,n\}\bigr\}$, and one is interested in finding incremental solutions with a low competitive ratio. 
In this framework, upper and lower bounds on the best-possible competitive ratio have been derived for various classes of set functions~$f$~\cite{BernsteinDisserGrossHimburg-20,DisserKlimmSchewiorWeckbecker-23,DisserWeckbecker-23,DisserWeckbecker/25}.

While the incremental maximization framework is a suitable model for situations where a new infrastructure needs to be constructed from scratch, it fails to capture scenarios where an existing infrastructure is upgraded to a new technological standard.
As an example, suppose that $E$ is a set of cell phone towers, and let $\gfunc(S)$ be the total amount of data transferrable with a subset $S \subseteq E$ of the towers using an old technological standard. 
The technology of the cell phone towers is to be updated to a new and non-interoperable standard; let~$\hfunc(S)$ be the total amount of data transferable with a subset $S \subseteq E$ of the towers using the new standard.\footnote{In a similar vein, one may consider a pipeline network being upgraded from operation with natural gas ($g$) to operation with hydrogen ($h$); this is our motivation to use $g$ for the function evaluating the benefit from the old standard and $h$ for the function evaluating the benefit from the new standard.}
In this context it is natural to ask in which order the cell towers should be upgraded in order to ensure that the amount of data transferrable in total across both standards is high at every point in time.
To phrase and answer this question formally, in this paper, we introduce a novel \emph{incremental--decremental} optimization framework that generalizes incremental maximization. 

In analogy to the incremental maximization framework, we call an ordering $\pi = (e_1,\dots,e_n)$ of the elements an incremental--decremental solution. Such an order increments the set of elements using the new technology while decrementing the set of elements using the old technology, hence the name.
For two arbitrary set functions $\hfunc,\gfunc \colon 2^E \to \R_{\geq 0}$, let $f \colon 2^E \to \R_{\geq 0}$ be defined as $f(S) \coloneqq \hfunc(S) + \gfunc(E\setminus S)$.
Note that in the application above $f(S)$ is equal to the total amount of data transferable using both standards whence the subset $S \subseteq E$ of cell phone towers has been upgraded.
As before, let $\OPT_k \coloneqq \max \bigl\{ f(S) : S \subseteq E \text{ with } |S|=k \bigr\}$.
The \emph{competitive ratio} of the incremental--decremental solution $\pi$ is then defined as
\begin{align*}
\rho \coloneqq \inf \bigl\{ r \in \R_{\geq 0} : r f(S_k) \geq \OPT_k \text{ for all } k \in \{1,\dots,n\}\bigr\}.
\end{align*}
The competitive ratio of an algorithm is the supremum over the competitive ratios of its solutions.
Likewise the competitive ratio of a problem is the infinum over the competitive ratios of an algorithm for the problem. 
We say that an algorithm or a solution is $\rho$-\emph{competitive} if its competitive ratio is at most $\rho$.

Since the notions of an incremental--decremental solution for $\hfunc$ and $\gfunc$, and that of an incremental solution for $f$ defined by $f(S) = \hfunc(S) + \gfunc(E\setminus S)$ coincide, one may wonder why it is insightful to introduce incremental--decremental solutions in the first place.
The reason for this stems from the fact that the incremental--decremental framework allows to impose natural properties of the functions $\hfunc$ and $\gfunc$ that may not be phrased so conveniently for the combined function $f$. 
For instance, in the above example of upgrading the technology of cell phone towers, it is natural to assume that both the function $\hfunc$ and $\gfunc$ are monotone.
However, $f(S) = \hfunc(S) + \gfunc(E \setminus S)$ is then generally non-monotone and thus known results for incremental maximization (that all require monotonicity) do not apply.
As a consequence, we will see that the best-possible competitive ratio in the incremental--decremental setting is strictly worse than in the incremental setting for monotone submodular functions.
A general treatment of the incremental maximization problem with non-monotone objective function~$f$ however does not yield meaningful results as it is unbounded even if $f$ is submodular; see~\Cref{prop:inc_max_unbounded} in \Cref{app:inc_max_unbounded}.

\subsection{Our Results}

We give polynomial-time algorithms computing an in\-cre\-men\-tal--de\-cre\-men\-tal solution whose performances are parametrized by the generic submodularity ratio $\gamma \in (0,1]$ and the curvature~$c \in (0,1]$ of the functions $g$ and $h$; these properties of set functions are formally defined in \Cref{def:generic_submodularity_ratio,def:curvature}.

As a warm-up, we examine a simple randomized algorithm that takes uniformly at random the greedy order of the elements in $E$ regarding $h$ or the reverse greedy order with respect to~$g$.
For this randomized algorithm, an upper bound of $2c\frac{\e^{c\gamma}}{\e^{c\gamma}-1}$ on the competitive ratio follows rather straightforwardly from known results for incremental maximization.

\newcolumntype{C}[1]{>{\centering\arraybackslash}p{#1}} 

\renewcommand{\arraystretch}{1.4}

\begin{table}[tb]
    \centering
    \begin{tabular}{p{3.0cm}C{1.4cm}@{}p{1.2cm}C{2.2cm}@{}p{1.5cm}C{1.4cm}@{}p{0.5cm}}
        \toprule
        & \multicolumn{4}{c}{\bf incremental--decremental} &\multicolumn{2}{c}{\bf incremental} \\
          & \multicolumn{2}{c}{lower bound} &  \multicolumn{2}{c}{upper bound} & \multicolumn{2}{c}{upper bound}\\  
        \midrule 
        
        \multirow{2}{*}{\textbf{general}} &
            $1+\frac{c}{2-c}$
        & 
        \scriptsize{(Thm~\ref{thm:lower})} &
        \cellcolor{lightgray}$2c\frac{\e^{c\gamma}}{\e^{c\gamma}-1}$
        &
        \cellcolor{lightgray}\scriptsize{(Cor~\ref{cor:random})} 
        & 
        \multirow{2}{*}{$c\frac{\e^{c\gamma}}{\e^{c\gamma}-1}$} & 
\multirow{2}{*}{\scriptsize{\cite{bian2017guarantees}}} \\[5pt]
      & $\frac{2}{3}+\frac{1}{3\gamma}$ &
      \scriptsize{(Thm~\ref{thm:lower})} & 
      $\frac{1}{\gamma}(1+c\frac{\e^c}{\e^c-1})$ & \scriptsize{(Thm~\ref{thm:main1})} \\
        \textbf{submodular} & $2$ & \scriptsize{(Thm~\ref{thm:lower})} & $ 1\!+\!\frac{\e}{\e\!-\!1} \!\approx\! 2.58$ & \scriptsize{(Cor~\ref{cor:main1})} & $\frac{\e}{\e-1}$ &\scriptsize{\cite{nemhauser1978analysis}} \\
        \textbf{gross substitute} &  $5/4$ & \scriptsize{(Thm~\ref{thm:lower})} & $2$ & \scriptsize{(Thm~\ref{thm:competitve_algo_gross_substitutes})} & $1$ & \scriptsize{\cite{leme2017gross}}\\
        \bottomrule\\
    \end{tabular}
    \caption{Our lower and upper bounds on the competitive ratio of incremental--decremental maximization for different classes of functions~$\hfunc$ and~$\gfunc$ with curvature $c$ and generic submodularity ratio~$\gamma$; results with shaded background are for randomized algorithms. Known upper bounds on the competitive ratios for incremental maximization (corresponding to the case $\gfunc \equiv 0$) are for comparison.\label{tab:overview}}
\end{table}

As our main contribution, we then introduce a deterministic algorithm that can be viewed as a way to derandomize the randomized approach. Instead of following the greedy order either only for $g$ or only for $h$, the algorithm builds both greedy orders partially at the same time depending on which promises a larger incremental gain in the current step. 
For this deterministic algorithm that we term \emph{double-greedy}, we prove the following competitive ratio which improves on the randomized competitive ratio for large range of values of the curvature~$c \in (0,1]$ and the generic submodularity ratio~$\gamma \in (0,1]$. Specifically, it provides a better competitive ratio as long as $\gamma \geq 0.54$ which includes in particular all submodular functions.

 \begin{restatable}{theorem}{mainupperbound}
 \label{thm:main1}
        For monotone functions $\gfunc, \hfunc \colon 2^E \to \R_{\geq 0}$ with curvature~$c \in (0,1]$ and generic submodularity ratio~$\gamma \in (0,1]$, the double-greedy algorithm is $\frac{1}{\gamma}\big(1+c\frac{\e^c}{\e^c-1}\big)$-competitive.
\end{restatable}

Using that for a submodular function, we have $\gamma = 1$, $c \in (0,1]$, and that the function $c \mapsto c \frac{\e^c}{\e^c-1}$ is non-decreasing in $c$, we obtain the following direct corollary of \Cref{thm:main1}.

\begin{corollary}\label{cor:main1}
For monotone submodular functions $\gfunc, \hfunc$, the double-greedy algorithm is $\big(1+\frac{\e}{\e-1}\big)$-competitive.
\end{corollary}

We further consider gross substitute functions, a subclass of submodular functions.
For this class of functions, we show the following result.

\begin{restatable}{theorem}{thmuppergs}
\label{thm:competitve_algo_gross_substitutes}
For monotone and gross substitute functions $\gfunc, \hfunc : 2^E \to \R_{\geq 0}$ the double-greedy algorithm is $2$-competitive.
\end{restatable}

We complement our results by the following lower bounds.

\begin{restatable}{theorem}{lowerbounds}
\label{thm:lower}
Let $\gfunc,\hfunc \colon 2^E \to \R_{\geq 0}$ be monotone functions. Then no algorithm has a better competitive ratio than\vspace{-2pt}
\begin{itemize}
\item[$\cdot$] $\smash{1+\frac{c}{2-c}}$ if $\gfunc$ and $\hfunc$ are submodular with curvature $c \in [0,1]$; 
\item[$\cdot$] $\smash{\frac{2}{3}+\frac{1}{3\gamma}}$ if $\gfunc$ and $\hfunc$ have generic submodularity ratio $\gamma \in (0,1]$;
\item[$\cdot$] $5/4$ if $\gfunc$ and $\hfunc$ are gross substitute.
\end{itemize}
\end{restatable}
For an overview of known results as well as our contributions, we refer to Table~\ref{tab:overview}.

\subsection{Related Work}

Our work on incremental--decremental maximization directly relates to literature on incremental maximization where the latter corresponds to the special case that $\gfunc \equiv 0$.
The natural greedy algorithm for cardinality constrained maximization inherently produces an incremental solution.
Its analysis for monotone submodular functions by Nemhauser et al.~\cite{nemhauser1978analysis} translates to a competitive ratio of $\frac{\e}{\e-1}$.
This approximation guarantee is best-possible among all polynomial algorithms, unless $\mathsf{P} = \mathsf{NP}$, as shown by Feige \cite{Feige1998threshold}. 
It is interesting to note that in a setting where both $h$ and $g$ are monotone submodular, our results imply that the best-possible competitive ratio is in the interval $[2,1+\frac{\e}{\e-1}]$ thus separating the competitive ratios achievable in the incremental and the incremental--decremental settings.
For monotone gross substitute functions, a result of Paes~Leme \cite{leme2017gross} implicitly yields an incremental solution that is $1$-competitive. 
In contrast, our results imply that in a setting where both $g$ and $h$ are gross substitute, the best-possible competitive ratio for incremental--decremental maximization is in the interval $[5/4,2]$ exhibiting that also for this class of functions incremental--decremental maximization is strictly more difficult than incremental maximization.
For submodular functions with curvature~$c$, Conforti et al.~\cite{conforti1984submodular} provided a more fine-grained analysis that yields a competitive ratio of~$c\frac{\e^c}{\e^c-1}$ for the greedy algorithm. 
For functions with submodularity ratio~$\SR$, Das et al.~\cite{das2018approximate} showed that the greedy algorithm yields a $\smash{\frac{\e^{\SR}}{\e^{\SR}-1}}$-approximation. 
Bian et al.~\cite{bian2017guarantees} combined the analysis of Conforti et al.~and Das et al.~and showed that for functions of curvature~$c$ and submodularity ratio~$\SR$, the greedy algorithm yields a $\smash{c\frac{\e^{c\SR}}{\e^{c\SR}-1}}$-approximation.
Disser et al.~\cite{DisserWeckbecker-23} considered the more general setting of $\SR$-$\alpha$-augmentable functions and showed that the greedy algorithm is a $\smash{\frac{\alpha}{\SR} \cdot \frac{\e^{\alpha}}{\e^{\alpha}-1}}$-approximation. 
This class of functions contains all functions with submodularity ratio~$\SR$ as well as all $\alpha$-augmentable functions. 
While the results above are stated in terms of the submodularity ratio $\SR$, they also hold for the generic submodularity ratio~$\gamma$ since $\SR \geq \gamma$.

Incremental maximization beyond the greedy algorithm has been considered by Bernstein et al.~\cite{BernsteinDisserGrossHimburg-20}. 
For accountable functions, they gave a $(\varphi + 1)$-competitive solution where $\varphi \approx 1.62$ is the golden ratio. 
Here, a function $f \colon 2^E \to \R_{\geq 0}$ is called accountable if for all $S \subseteq E$ there is $s \in S$ such that $f(S\setminus\{s\})\geq f(S) - f(S)/|S|$. 
They further showed that no deterministic algorithm can provide a competitive ratio better than~$2.18$ in this setting. 
Disser et al.~\cite{DisserKlimmSchewiorWeckbecker-23} gave an improved lower bound of $2.24$ and gave a randomized upper bound of~$1.77$.

Goemans and Unda~\cite{Goemans2017approximating} considered a different setting of incremental maximization where, instead of the worst-case ratio over all time steps, the ratio of the sum over all time steps is to be maximized, and obtained bounds for special functions containing rank functions of independence systems. 

The double-greedy algorithm that we consider is reminiscent of double-greedy algorithms for non-monotone submodular maximization. Buchbinder et al.~\cite{Buchbinder2015tight} used this approach and obtained a deterministic $3$-approximation and a randomized $2$-approximation for this problem. 
Feige et al.~\cite{Feige2007maximizing} showed earlier that no better approximation than~$2$ is possible with polynomially many oracle calls, implying that the randomized $2$-approximation of Buchbinder et al.~cannot be improved.
Non-montone submodular function maximization has also been considered under knapsack, matroid, and other constraints~\cite{ChekuriVZ14,FadaeiFS11,KulikST13,LeeMNS10}.
    
   \section{A Simple Randomized Algorithm}\label{sec:random}

    To begin with, we state a first observation that leads to a naive randomized algorithm for the incremental--decremental problem.
    For a randomized algorithm the returned incremental--decremental solution $\pi$ and, hence, all prefixes $S_k$ are random variables; such an algorithm is called $\rho$-competitive if 
    $\rho\,\mathbb{E}[f(S_k)] \geq \OPT_k$ for all $k$.
    
    \begin{restatable}{proposition}{randomizedtworhocompetitive} \label{prop:2rho}
        Let $g,h\colon 2^E\rightarrow \mathbb R_{\geq 0}$.
		If there is a $\rho$-competitive algorithm both for the incremental problem for $g$ and for the incremental problem for~$h$, then there is a $2\rho$-competitive randomized algorithm for the incremental--decremental maximization problem for~$f\colon 2^E\rightarrow \mathbb R_{\geq 0}$, $S\mapsto h(S)+g(E\setminus S)$.
	\end{restatable}
    \begin{proof}
        The randomized algorithm takes with probability~$\frac{1}{2}$ the solution of the incremental algorithm for~$h$ and with probability $\frac{1}{2}$ the reverse order of the incremental solution for~$g$.
        Let $\pi^h$ and $\pi^g$ be $\rho$-competitive solutions to the incremental problem with respect to the set functions $h$ and $g$, respectively.
        Let $k\in \{1,\dots, n\}$.
        Denote with $S_k^h\subseteq E$ the subset with the first $k$ elements of $\pi^h$ and with $S_k^g\subseteq E$ the subset containing the last $k$ elements of~$\pi^g$.
        Let~$\OPT_k^h$, $\OPT_k^g$ and $\OPT_k^f$ be the maximum values of a subset with $k$ elements for objective functions $h$, $g$ and $f$, respectively.

        The expected value of the first $k$ elements of the order returned by the randomized algorithm for the incremental--decremental problem is then given by, 
        $\frac{1}{2}(f(S_k^h)+f(S_k^g))\geq \frac{1}{2}(h(S_k^h)+g(E\setminus S_k^g)) \geq \frac{1}{2\rho}(\OPT_k^h+\OPT_{n-k}^g) \geq \frac{1}{2\rho}\OPT_k^f,$ 
        where in the second inequality we used $\rho$-competitiveness of $\pi^h$ and $\pi^g$.
	\end{proof}
    
    To obtain concrete bounds on the competitiveness of this algorithm, we use the 
    upper bounds on the competitive ratio of incremental maximization for functions with bounded curvature $c\in(0,1]$ and bounded generic submodularity ratio $\gamma\in(0,1]$ listed in \Cref{tab:overview}.
    The generic submodularity ratio is a way to relax submodularity and allows to bound the competitive ratio for functions that are in some sense close to submodular.
    In the following definitions, we set $\min \emptyset \coloneqq 1$.
    For the marginal increase in~$f\colon 2^E \to \R_{\geq 0}$ of an element~$e \in E$ or of a set~$T \subseteq E$ with respect to~$S \subseteq{E}$, we introduce the notations
    $f(e \mid S) \coloneqq f(\{e\} \cup S) -f(S)$ and $f(T \mid S) \coloneqq f(T \cup S) -f(S)$.
    
    \begin{definition}[Gong et al.~\cite{gong2021maximize}] \label{def:generic_submodularity_ratio}
        A monotone set function $f \colon 2^E \rightarrow \R_{\geq0}$ has \emph{generic submodularity ratio} 
        \[
            \gamma \coloneqq \min \biggl\{ \frac{f(e \mid A)}{f(e \mid A \cup B)} : A,B \in 2^E, e \in E \text{ with } f(e\mid A\cup B) > 0\biggr\}.
        \]
    \end{definition}

Intuitively, the generic submodularity ratio interpolates between monotone functions with $\gamma=0$ and submodular functions with $\gamma=1$.
To see this, note that for a submodular function, we have $f(e \mid A) \geq f(e \mid A \cup B)$ for all $A,B \in 2^E$ and $e \in E$, so that the minimum is attained for $B = \emptyset$ and the submodularity ratio is $\gamma = 1$. 
Das and Kempe~\cite{das2018approximate} introduced the stronger notion of the \emph{submodularity ratio}
\begin{align*}
            \SR \coloneqq \min \bigg\{\frac{\sum_{b \in B} f(b \mid A)}{f(B\mid A)} : A,B \in 2^E \text{ with } f(B \mid A)>0 \biggr\}
\end{align*}
    and showed that the greedy algorithm is $\smash{\frac{\e^{\SR}}{\e^{\SR}-1}}$-competitive for $\gfunc \equiv 0$.
    Observe that the generic submodularity ratio $\gamma$ is always smaller or equal to $\SR$, since for all subsets $B=\{b_1,\dots, b_k\}, A \subseteq E$ with $f(B \mid A)>0$,
    \[
        \frac{\sum_{b \in B} f(b \mid A)}{f(B\mid A)}
        = \frac{\sum_{i=1}^k f(b_i \mid A)}{\sum_{i=1}^k f(b_i \mid A \cup \{b_1, \dots, b_{i-1}\})}
        \geq \frac{\sum_{i=1}^k f(b_i \mid A)}{ \gamma^{-1} \sum_{i=1}^k f(b_i \mid A)} =\gamma.
    \]
    Thus, their bound implies also $\frac{\e^{\gamma}}{\e^{\gamma}-1}$-competitiveness in terms of the generic submodularity ratio in the incremental setting.

    The curvature $c\in[0,1]$ measures how close a function is to being supermodular.
    A function is supermodular if it has curvature $c=0$.
    
    \begin{definition}[Bian et al.~\cite{bian2017guarantees}] \label{def:curvature}
        A monotone set function $f\colon 2^E \rightarrow \R_{\geq0}$ has \emph{curvature}
        \[
            c:= 1- \min \biggl\{ \frac{f(e\mid A \cup B)}{f(e\mid A)} : A,B \in 2^E, e \in E \setminus B \text{ with } f(e\mid A)>0\biggr\}.
        \]
    \end{definition}
    
    Bian et al.~\cite{bian2017guarantees} showed that the greedy algorithm is $c\frac{\e^{c\SR}}{\e^{c\SR}-1}$-competitive for $\gfunc \equiv 0$. 
    For submodular functions the notion of curvature was already introduced by Conforti and Cornu{\'e}jols~\cite{conforti1984submodular}.
    In this case, the curvature simplifies to 
    \begin{align} \label{eq:def_curvature_submodular}
        1- \min_{e \in E} \frac{f(e \mid E \setminus \{e\})}{f(\{e\})}
    \end{align}
    because of $f(e \mid A \cup B) \geq f(e \mid E \setminus \{e\})$ and $f(e \mid A) \leq f(\{e\})$.
    We observe that submodular functions with curvature zero are modular since they are both submodular and supermodular. 

    The upper bound on the competitive ratio of the incremental problem together with~\Cref{prop:2rho} translates to a first upper bound on the competitive ratio of incremental--decremental maximization.

    \begin{corollary}\label{cor:random}
        Let $\gfunc,\hfunc\colon 2^E \rightarrow \R_{\geq0}$ be monotone with generic submodularity ratio $\gamma \in (0,1]$ and curvature~$c \in (0,1]$.
        The algorithm that returns uniformly at random the greedy order of~$E$ with respect to $h$ or the reverse greedy order with respect to $g$ is  $2c\frac{\e^{c\SR}}{\e^{c\SR}-1}$-competitive.
    \end{corollary}

    \section{A  Deterministic Algorithm}\label{sec:incdec}

    The double-greedy algorithm (see below) can be seen as a way to derandomize the randomized solution.
    Instead of only considering the elements in the greedy-order defined by $g$ or the inverse greedy-order defined by $h$, it instead builds partial greedy orders for both functions from both ends simultaneously. 
    Initially both the considered prefix and the considered suffix of the solution are empty.
    In each step, an element is selected that maximizes the maximum of the increment relative to the current prefix for $h$ and the increment relative to the current suffix for $g$.
    It is then appended from the right to the current prefix if the maximum is attained for $h$, and from the left to the current suffix if the maximum is attained for $g$. 
    This procedure is visualized in \cref{fig:double-greedy-step}.

    In the following, we denote by $\setF_k \coloneqq \setG_k \cup \setH_k$ the set of all elements that are added to the solution by the double-greedy algorithm up to step $k$.
    Also, we denote the complement of a set~$S$ in $E$ by $S^{\comp}\coloneqq E \setminus S$ and write $S_k^{\comp}\coloneqq (S_k)^{\comp}$ for indexed sets.
    Further, we generally let $n\coloneqq \abs{E}$ and let $\ffunc\colon 2^E \rightarrow \R_{\geq 0}$ denote the objective function of the incremental--decremental maximization problem, 
    i.e., $\ffunc(S) \coloneqq \hfunc(S) +\gfunc(S^{\comp})$.
    The double-greedy algorithm is symmetric with respect to the roles of~$\gfunc$ and~$\hfunc$ with the slight complication that the handling of ties between~$\gfunc$ and~$\hfunc$ is inherently asymmetric. 
    In order to still exploit symmetry, we parameterize the algorithm by a comparison operator $\prec \ \in \{ \leq, < \}$ that governs whether it favors adding elements to $\setG_i$ or $\setH_i$ in case of a tie and consider both variants simultaneously.
    The following proposition makes precise how we can exploit symmetry in our proofs.
    In particular, it suffices to show competitiveness of the solution up to cardinality~$|\setH_n|$.
    
    \begin{restatable}{proposition}{propositiononsymmetry} \label{prop:symmetry_handling}
        Let $\mathcal{C} \subseteq \{ \hfunc \colon 2^E \rightarrow \R_{\geq0} \text{ monotone}\}$. Assume that for all functions $\gfunc, \hfunc \in \mathcal{C}$, for all $\prec \ \in \{\leq, <\}$, for all $k\leq n$ and for all subsets $S \subseteq E$ with $\abs{S}=\abs{\setH_k}$ the double-greedy algorithm ensures that
        \begin{align} \label{eq:rho_competitiveness}
            \rho \ffunc(\setH_k) \geq \ffunc(S),
        \end{align}
        for some~$\rho \geq 1$.
        Then, the double-greedy algorithm is $\rho$-competitive for all functions $\gfunc, \hfunc \in \mathcal{C}$.
    \end{restatable}
    \begin{proof}
        Let $\pi = (e_1,\dots,e_n)$ denote the ordering computed by the double-greedy algorithm.
        Then, the double-greedy algorithm is $\rho$-competitive if, for all $\ell \leq n$ and all sets $S \subseteq E$ with~$\abs{S}=\ell$, we have $\rho \ffunc(\{e_1, \dots, e_{\ell}\}) \geq \ffunc(S)$.
        Equivalently, 
        \begin{align}
            \rho \ffunc(\setH_k) \geq \ffunc(S)\label{eq:symmetry1}
        \end{align}
        for all sets $S \subseteq E$ with $\abs{S}=\abs{\setH_k}$ and 
        \begin{align}
            \rho \ffunc(\setG_k^{\comp}) \geq \ffunc(S)\label{eq:symmetry2}
        \end{align}
        for all sets $S \subseteq E$ with $\abs{S}=\abs{\setG^{\comp}_k}$.
        By assumption, \eqref{eq:symmetry1} holds and it remains to prove \eqref{eq:symmetry2}.
        
        The double-greedy algorithm is symmetric in the following sense. 
        If we swap the functions~$\gfunc$ and $\hfunc$ in the input and change the choice of $\prec \ \in \{\leq,<\}$, the double-greedy algorithm produces exactly the reversed ordering of elements, provided that ties are handled consistently.
        
        Denote the modified input to the double-greedy algorithm by $\gfunc'\coloneqq\hfunc$, $\hfunc'\coloneqq\gfunc$ and $\prec' \ \in \{\leq, <\}$ with $\prec' \neq \prec$. 
        Furthermore, we write $\ffunc'(S)\coloneqq\hfunc'(S)+\gfunc'(S^{\comp})$. 
        Let~$\setH'_i$ and $\setG'_i$ be the sets~$\setH_i$ and~$\setG_i$ computed by the double-greedy algorithm for the modified input $\gfunc'$, $\hfunc'$ and $\prec'$.
        Similarly, let $\setH_i$ and $\setG_i$ be the sets computed by the algorithm on the input $\gfunc$, $\hfunc$ and $\prec$.
        We show by induction on~$i$ that $\setH'_i=\setG_i$ and $\setG'_i=\setH_i$.
        This trivially holds for~$i=0$.
        
        Consider iteration $i+1$. 
        By induction, $\gfunc'(e \mid \setG'_i)=\hfunc(e \mid \setH_i)$ and~$\hfunc'(e \mid \setH'_i)=\gfunc(e \mid \setG_i)$. 
        Thus, for both inputs the same choice for $e^*_{i+1}$ is possible in the double-greedy algorithm.
        The if-condition in the double-greedy algorithm with input $\gfunc'$, $\hfunc'$ and~$\prec'$ is $\hfunc'(e^*_{i+1} \mid \setH'_{i}) \prec' \gfunc'(e^*_{i+1} \mid \setG'_i)$, which is equivalent to $\gfunc(e^*_{i+1} \mid \setG_{i}) \prec' \hfunc(e^*_{i+1} \mid \setH_i)$. 
        Hence, the condition is true if and only if the if-condition in the double-greedy algorithm with input $\gfunc$, $\hfunc$ and $\prec$ given by~$\hfunc(e_{i+1}^* \mid \setH_{i}) \prec \gfunc(e_{i+1}^* \mid \setG_{i})$ is false.
        Thus,~$e^*_{i+1}$ is added to $\setH'_i$ if and only if it is added to~$\setG_i$. 
        Similarly for $\setG'_i$ and $\setH_i$, which implies $\setH'_{i+1}=\setG_{i+1}$ and $\setG'_{i+1}=\setH_{i+1}$.
        Therefore, the double-greedy algorithm on $\gfunc'$, $\hfunc'$ and $\prec'$ computes the reversed ordering.
        
        We can now apply $(\ref{eq:rho_competitiveness})$ for the double-greedy algorithm on~$\gfunc'$, $\hfunc'$ and $\prec'$ to obtain
        \vspace{-.2cm}
        \begin{align*}
            \rho \ffunc(\setG^{\comp}_k) 
            &= \rho (\hfunc(\setG_k^{\comp})+ \gfunc(\setG_k)) 
            = \rho ( \gfunc'((\setH'_k)^{\comp}) + \hfunc'(\setH'_k) ) \\ 
            &= \rho \ffunc'(\setH'_k) 
            \overset{(\ref{eq:rho_competitiveness})}{\geq} \ffunc'(S^\comp) 
            = \ffunc(S).\tag*{\raisebox{-0.5ex}{\qedhere}}
        \end{align*}
    \end{proof}

    \begin{algorithm}[tb]
        \caption{\label{alg:double-greedy}double-greedy algorithm}
        \Input{$\gfunc,\hfunc \colon 2^E \rightarrow \R_{\geq0}$ monotone, $\prec \ \in \{\leq, <\}$}
        \Output{ordering $\pi=(e_1, \dots, e_n)$ of the elements in $E$}
        \vspace{-0.75em}\hrulefill\vspace{0.25em}
        
            $\setG_0 \gets \emptyset$; $\setH_0 \gets \emptyset$ \;
            \For{$i \gets 1,\dots, n$}
            {
                $e_i^* \gets \argmax_{e \in E \setminus (\setG_{i-1} \cup \setH_{i-1})} \{ \max\{\gfunc(e \mid \setG_{i-1}), \hfunc(e \mid \setH_{i-1})\}\}$
                
                \uIf{$\hfunc(e_i^* \mid \setH_{i-1}) \prec \gfunc(e_i^* \mid \setG_{i-1})$}
                {
                    $\setG_{i}\gets \setG_{i-1} \cup \{e_i^*\}$\;
                    $\setH_{i} \gets \setH_{i-1}$\;
                    $e_{n-|\setG_{i}|+1}\gets e_i^*$\;
                }
                \Else
                {
                    $\setH_{i} \gets \setH_{i-1} \cup \{e_i^*\}$\;
                    $\setG_{i}\gets \setG_{i-1}$\;
                    $e_{|\setH_{i}|}\gets e_i^*$\;
                }
            }
            \Return $\pi \gets (e_1, \dots, e_n)$
    \end{algorithm}

    \begin{figure}
        \centering
    	\begin{tikzpicture}[scale=0.35]
    		\node at (0,7) {\makebox[0pt][r]{$\pi=$}$(e_1, \dots, e_k, \quad ? \quad, e_l, \dots, e_n)$};
    		\draw[decorate, decoration={brace, mirror, amplitude=6pt}] 
    		(-5.75,6.5) -- (-1.75,6.5) node[midway, below=6pt] {\footnotesize $H_{i-1}$};
    		\draw[decorate, decoration={brace, mirror, amplitude=6pt}] 
    		(1.75,6.5) -- (5.75,6.5) node[midway, below=6pt] {\footnotesize $G_{i-1}$};
    		
    		\draw (0,0) ellipse [x radius=3cm, y radius=1.5cm];
    		\node (elem) at (0,0) {$e^*_{i}$};
    		\node[right] at (3,0) {\footnotesize $E \setminus (H_{i-1} \cup G_{i-1})$};
    		
    		\draw[-{Stealth}, bend left] (elem) to node[left,pos=0.4] {\footnotesize $h(e^*_{i} \mid H_{i-1})\geq g(e^*_{i} \mid G_{i-1})$} (-1,6.5);
    		\draw[-{Stealth}, bend right] (elem) to node[right,pos=0.4] {\footnotesize $h(e^*_{i} \mid H_{i-1})< g(e^*_{i} \mid G_{i-1})$} (1,6.5);
    	\end{tikzpicture}
        \caption{In step $i$ of the double-greedy algorithm (with $\prec=<$), the element $e^*_{i}$ maximizing $\max\{g(e \mid G_{i-1}), h(e \mid H_{i-1})\}$ among the remaining elements $e \in E\setminus(G_{i-1}\cup H_{i-1})$ is added to the left or right depending on which marginal increase, $g(e^*_i \mid G_{i-1})$ or $h(e^*_i \mid H_{i-1})$, is higher.}
        \label{fig:double-greedy-step}
    \end{figure}
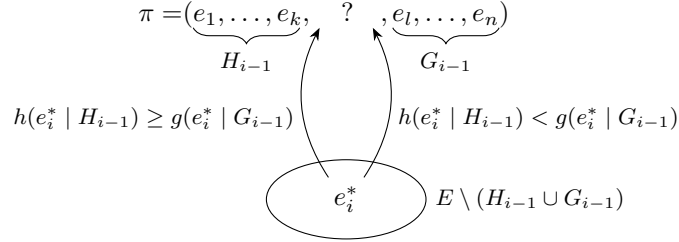
    From now on, we assume that $\hfunc(\emptyset) = \gfunc(\emptyset) = 0$.
    This assumption is without loss of generality, since it suffices to show $\rho$-competitiveness for the functions $\gfunc'(S)\coloneqq \gfunc(S) - \gfunc(\emptyset)$ and $\hfunc'(S)\coloneqq \hfunc(S) - \hfunc(\emptyset)$, because of
    \begin{multline*}
      \ffunc(S) \coloneqq \hfunc(S) + \gfunc(S^{\comp})
      = \hfunc'(S) + \hfunc(\emptyset) + \gfunc'(S) + \gfunc(\emptyset) \\
        \leq \rho (\hfunc'(\setH_k) + \gfunc'(\setH_k^{\comp}))+\hfunc(\emptyset)+\gfunc(\emptyset)                     
        \leq \rho (\hfunc'(\setH_k) + \hfunc(\emptyset) + \gfunc'(\setH_k^{\comp}) +\gfunc(\emptyset))
        =\rho \ffunc(\setH_k),
    \end{multline*}
    for all~$k \leq n$ and $\abs{S}=\abs{H_k}$, which proves competitiveness by Proposition~\ref{eq:rho_competitiveness}.

    We conclude this subsection with a lemma which we will use to prove our upper bounds.
    
    \begin{restatable}{lemma}{lemmagenerallowerboundonopt} \label{lem:gen_lower_bound_for_double_greedy}
        Let $\gfunc, \hfunc \colon 2^E \to \R_{\geq 0}$ be monotone. After iteration $k$ of the double-greedy algorithm we have
         $\ffunc(\setH_k) \geq \sum^k_{i=1} \max\{\hfunc(e^*_i \mid \setH_{i-1}), \gfunc(e^*_i \mid \setG_{i-1}) \}$.
    \end{restatable}
    \begin{proof}
        By monotonicity of $\gfunc$, using $\setG_k \subseteq \setH_k^{\comp}$, we have
        \begin{align*}
            \ffunc(\setH_k) &= \hfunc(\setH_k) + \gfunc(\setH_k^{\comp})\geq \hfunc(\setH_k) + \gfunc(\setG_k) \\
            &= \sum_{i:e^*_i \in \setH_k} \hfunc(e^*_i \mid \setH_{i-1}) + \sum_{i:e^*_i \in \setG_k} \gfunc(e^*_i \mid \setG_{i-1}).
        \end{align*}
        By the definition of the algorithm we obtain
        \begin{align*}
            \ffunc(\setH_k) &\geq \!\!\sum_{i:e^*_i \in \setH_k}\!\! \hfunc(e^*_i \mid \setH_{i-1}) + \!\!\sum_{i:e^*_i \in \setG_k}\!\! \gfunc(e^*_i \mid \setG_{i-1})\\
            &\geq \sum_{i=1}^{k} \max\{\hfunc(e^*_i \mid \setH_{i-1}), \gfunc(e^*_i \mid \setG_{i-1} ) \}.\tag*{\raisebox{-2ex}\qedhere}
        \end{align*}
    \end{proof}
    
    \subsection{Bounded Generic Submodularity Ratio and Curvature}    
    In this section, we consider the case where $\gfunc$ and $\hfunc$ have bounded curvature $c$ and bounded generic submodularity ratio $\gamma$ which we have defined in \Cref{sec:random}.
    Both parameters, the generic submodularity ratio and the curvature characterize how much the marginal increase of an element is allowed to change on different subsets. 
    A function has generic submodularity ratio $\gamma$ if the marginal values of each element increase at most by a factor of $\gamma$, i.e., we have by \Cref{def:generic_submodularity_ratio} that 
    \begin{align} \label{eq:def_gen_sub_mod_ratio}
        f(e \mid S)\geq \gamma f(e \mid T) \qquad \text{ for } S \subseteq T \subseteq E, e \in E.
    \end{align}
    
    Likewise, the curvature measures how much the marginal increase decreases at most, i.e., \Cref{def:curvature} implies 
    \begin{align} \label{eq:def_curvature}
        (1-c)f(e \mid S)\leq f(e \mid T) \qquad \text{ for } S \subseteq T \subseteq E, e\in E \setminus T.
    \end{align}

    We proceed to analyze the double-greedy algorithm for functions with fixed generic submodularity ratio and curvature.
    Recall, that $\setF_m=\setH_m \cup \setG_m$ is the set of elements that are added to the solution up to step $m$.
    Let $\varphi_i \coloneqq \max \{\gfunc(e^*_i \mid \setG_{i-1}), \hfunc(e^*_i \mid \setH_{i-1})\}$ be the increment governing the choice of $e^*_i$ in each step of the algorithm. 
    We obtain a lower bound on $\varphi_i$ by considering the gap to the value of a larger solution $S\subseteq E$. 
    In particular, we will use this bound with $S$ being an optimum solution for $\ffunc$ of size $\abs{S}$.
    
    \begin{restatable}{lemma}{lemmaboundmarginalincrease} \label{lem:lower_bound_of_marginal_increase}
        Let $\gfunc,\hfunc\colon 2^E \rightarrow \R_{\geq0}$ be monotone with generic submodularity ratio $\gamma \in (0,1]$ and curvature~$c \in (0,1]$. For $S \subseteq E$ and $0 \leq m \coloneqq |S|$ the double-greedy algorithm ensures that
        \begin{align*}
                \varphi_{m+1} \geq \frac{c}{|S\setminus \setF_{m}|} \left(\gamma \hfunc(S) - \sum_{i=1}^{m} \varphi_{i} \right)+ \frac{1-c}{|S\setminus \setF_{m}|} \left( \gamma \hfunc(S) - \sum_{i:e^*_i \in \setF_m \cap S} \varphi_{i} \right).
        \end{align*}
    \end{restatable}
    \begin{proof}
        Let $\{s_1, \dots, s_k\} \coloneqq S\setminus F_m$. 
        Then, we have using $H_{m} \subseteq \{s_1, \dots, s_{i-1}\} \cup F_m$ that
        \begin{align} \label{eq:bound_sol_without_alg_submod}
            \hfunc(S \setminus \setF_m \mid \setF_m) = \sum_{i=1}^k \hfunc(s_i \mid \{s_1, \dots, s_{i-1}\} \cup F_m) \overset{(\ref{eq:def_gen_sub_mod_ratio})}{\leq} \sum_{s \in S \setminus \setF_{m}} \gamma^{-1} \hfunc(s \mid \setH_{m}).
        \end{align}
        Now, we use $\varphi_{m+1} \geq \hfunc(e \mid \setH_{m})$ for $e \notin F_m$ by the greedy choice of $e^*_{m+1}$ to bound
        \begin{align} \label{eq:bound_sol_without_alg_greedy_choice}
            \sum_{s \in S \setminus \setF_{m}} \gamma^{-1} \hfunc(s \mid \setH_{m}) \leq |S \setminus \setF_{m} | \ \gamma^{-1} \varphi_{m+1}.
        \end{align}
        Together, we can compute
        \begin{eqnarray}
            \label{eq:bounding_sol_with_algo_increments}
            \hfunc(S) 
            &=& \hfunc(F_m \cup S) - \hfunc(\setF_m \mid S)\nonumber\\
            &=& \hfunc(\setF_m) + \hfunc(S \setminus \setF_m \mid \setF_m) - \hfunc(\setF_m \mid S) \nonumber\\
            &\overset{(\ref{eq:bound_sol_without_alg_submod})}{\leq}& \hfunc(\setF_m) + \sum_{s \in S \setminus \setF_{m}} \gamma^{-1} \hfunc(s \mid \setH_{m}) - \hfunc(\setF_m \mid S) \nonumber\\[-0.5em]
            &\overset{(\ref{eq:bound_sol_without_alg_greedy_choice})}{\leq}& \hfunc(\setF_{m}) + |S \setminus \setF_{m} | \cdot \gamma^{-1} \varphi_{m+1} - \hfunc(\setF_m \mid S) \nonumber\\
            &=& \sum_{i=1}^{m} \hfunc(e^*_i \mid \setF_{i-1}) + |S \setminus \setF_{m} | \cdot \gamma^{-1} \varphi_{m+1} - \sum_{i=1}^{m} \hfunc(e^*_i \mid S \cup \setF_{i-1}) \nonumber
        \end{eqnarray}
        \begin{eqnarray}
            &\overset{(\ref{eq:def_curvature})}{\leq}& \sum_{i=1}^{m} \hfunc(e^*_i \mid \setF_{i-1}) + |S \setminus \setF_{m} | \cdot \gamma^{-1} \varphi_{m+1}  - \sum_{i:e^*_i \in \setF_m \setminus S} (1-c) \hfunc(e^*_i \mid \setF_{i-1}) \nonumber\\[-0.5em]
            &=& (1-c) \sum_{i:e^*_i \in \setF_m \cap S} \! \! \! \hfunc(e^*_i \mid \setF_{i-1}) + c \sum_{i=1}^{m} \hfunc(e^*_i \mid \setF_{i-1}) + |S \setminus \setF_{m}| \ \gamma^{-1} \varphi_{m+1} \nonumber\\[-0.5em]
            &\overset{(\ref{eq:def_gen_sub_mod_ratio})}{\leq}&  (1-c) \!\! \sum_{i:e^*_i \in \setF_m \cap S} \!\!\!\! \gamma^{-1} \hfunc(e^*_i \mid \setH_{i-1})
            + c \sum_{i=1}^{m} \gamma^{-1} \hfunc(e^*_i \mid \setH_{i-1}) + |S \setminus \setF_{m} | \ \gamma^{-1} \varphi_{m+1} \nonumber \\
            &\overset{c\in[0,1]}{\leq}& (1-c) \!\! \sum_{i:e^*_i \in \setF_m \cap S} \!\!\!\! \gamma^{-1} \varphi_{i} + c \sum_{i=1}^{m} \gamma^{-1} \varphi_{i} + |S \setminus \setF_{m} | \ \gamma^{-1} \varphi_{m+1}.
        \end{eqnarray} 
        Isolation of $\varphi_{m+1}$ in $(\ref{eq:bounding_sol_with_algo_increments})$ yields the desired bound.
    \end{proof}
    
    We combine the bounds derived in \Cref{lem:lower_bound_of_marginal_increase} on the increments $\varphi_{m+1}$ in each step of the algorithm to bound the value of every solution by the sum of those increments.
    An optimum solution of each size is particularly bounded as follows.

    \begin{restatable}{lemma}{lemmaboundoptwithsumofincrements} \label{lem:double_greedy_bound_on_opt_submod_ratio_and_curvature}
        Let $\gfunc, \hfunc \colon 2^E \rightarrow \R_{\geq0}$ be monotone 
        with generic submodularity ratio $\gamma \in (0,1]$ and curvature~$c \in (0,1]$. 
        For $S \subseteq E$ and 
        $k \coloneqq |S|$, the double-greedy algorithm ensures that
        $
           \smash{\frac{1}{c}\left(1-\frac{1}{\e^{c}}\right)\gamma \hfunc(S) \leq \sum_{i=1}^{k} \varphi_{i}}.
        $
    \end{restatable}
    \begin{proof}
        For $0 \leq m \leq k$ let $k_m \coloneqq |S \setminus \setF_m|$ and let $\chi_i$ indicate whether $e^*_i$ is part of the solution $S$, i.e. $\chi_i=1$ if~$e^*_i \in S$ and~$\chi_i=0$ otherwise. 
        By definition, we have
        \begin{align} \label{eq:reference_releation_of_k}
            k_{m+1}=k_m-\chi_{m+1}.
        \end{align}
        We denote the difference between the sum of the increments in each step of the algorithm to the value $\gamma \hfunc(S)$, which we seek to bound, by $d_m \coloneqq \gamma \hfunc(S) - \sum_{i=1}^{m} \varphi_{i}$.
        Likewise, for the sum of only those increments belonging to elements contained in~$S$, we denote the difference by~$d^*_m \coloneqq \gamma \hfunc(S) - \sum_{i=1}^{m} \chi_i \varphi_{i}$.
        For~$d_m\leq0$, the statement of the lemma follows directly since we have $\frac{1}{c}\left(1-\e^{-c}\right)\gamma \leq 1$ for~$c,\gamma \in (0,1]$. Thus, we assume in the following that~$d_m>0$, which implies $d^*_m > 0$.

        We show by induction on $0 \leq m \leq k$ that
        \begin{align} \label{eq:induction_claim_submod_and_curvature_analysis}
            \bigg[ \left(1-\frac{c}{k} \right)^{k}&-1+c \bigg] \gamma \hfunc(S) \nonumber\\
            &\geq \left(1-\frac{c}{k} \right)^{k-m} c \, d_{m} - \left(1-\left(1- \frac{c}{k}\right)^{k-m}\right) \frac{(1-c)k}{k_{m}}d^*_{m}.
        \end{align}

        For $m=0$, $(\ref{eq:induction_claim_submod_and_curvature_analysis})$ holds with equality since by $d_0=d^*_0=\gamma \hfunc(S)$ and $k_0=\abs{S\setminus F_0}=k$ we have 
        \begin{align*}
            \left(1-\frac{c}{k} \right)^{k} c \, d_{0} &- \left(1-\left(1- \frac{c}{k}\right)^{k}\right) \frac{(1-c)k}{k_{0}}d^*_{0} \\
            &=\left[\left(1-\frac{c}{k} \right)^{k} c - \left(1-\left(1- \frac{c}{k}\right)^{k}\right) (1-c)\right] \gamma \hfunc(S)\\
            &= \left[ \left(1-\frac{c}{k} \right)^{k} -1 +c \right] \gamma \hfunc(S).
        \end{align*}
        In the following, we assume the statement holds for $m\geq 0$ to prove it for $m+1$.

        Inserting $d_{m+1}=d_{m}-\varphi_{m+1}$ and $d^*_{m+1}=d^*_m-\chi_{m+1}\varphi_{m+1}$ in the right hand side of $(\ref{eq:induction_claim_submod_and_curvature_analysis})$ for $m+1$ yields
        \begin{align} \label{eq:isolate_varphi_in_indunction}
            \bigg(1\ - & \ \frac{c}{k} \bigg)^{k-m-1} c\, d_{m+1} - \left(1-\left(1- \frac{c}{k}\right)^{k-m-1}\right) \frac{(1-c)k}{k_{m+1}}d^*_{m+1} \nonumber \\
            &= \left(1-\frac{c}{k} \right)^{k-m-1} c \, d_{m} - \left(1-\left(1- \frac{c}{k}\right)^{k-m-1}\right) \frac{(1-c)k}{k_{m}}d^*_{m} \nonumber \\
            &- \left[ \left(1-\frac{c}{k} \right)^{k-m-1} c - \bigg(1-\left(1- \frac{c}{k}\right)^{k-m-1}\bigg)  \frac{(1-c)k}{k_{m+1}} \chi_{m+1} \right] \varphi_{m+1}.
        \end{align}
        To apply \Cref{lem:lower_bound_of_marginal_increase}, we show non-negativity of the factor before $-\varphi_{m+1}$. 
        For this, we use Bernoulli's inequality which states that
        \begin{align} \label{eq:bernoullis_inequality}
            \left(1-x \right)^{n} \geq 1-nx \qquad \text{for }x \geq -1,n \in \N.
        \end{align}
        With $c\in[0,1]$ and $\frac{c}{k}\leq1$ we obtain non-negativity of the factor before $-\varphi_{m+1}$ given by
        \begin{eqnarray} \label{eq:non_negativity_of_factor_before_increment}
            \left(1-\frac{c}{k} \right)^{k-m-1} &c - \bigg(1-&\left(1- \frac{c}{k}\right)^{k-m-1}\bigg)  \frac{(1-c)k}{k_{m+1}} \chi_{m+1} \nonumber\\
            &\overset{\chi_{m+1}\leq1}{\geq}&  \left(1-\frac{c}{k} \right)^{k-m-1} \left(c + \frac{(1-c)k}{k_{m+1}} \right)- \frac{(1-c) k}{k_{m+1}} \nonumber\\
            &\overset{(\ref{eq:bernoullis_inequality})}{\geq}& \left(1-\frac{(k-m-1)c}{k}\right)\left(c + \frac{(1-c)k}{k_{m+1}} \right)-  \frac{(1-c)k}{k_{m+1}} \nonumber\\
            &=& c - \frac{(k-m-1)c^2}{k} -\frac{(k-m-1)c(1-c)}{k_{m+1}} \nonumber\\
            &\overset{k_{m+1} \leq k}{\geq}& c - \frac{k-m-1}{k_{m+1}}c \nonumber\\
            &\geq& 0,
        \end{eqnarray}
        where we use $k_{m+1} \geq k-m-1$ in the last step.
        By \Cref{lem:lower_bound_of_marginal_increase} and the definitions of~$k_m$,~$d_m$ and~$d^*_m$, we have
        \begin{align} \label{eq:bound_on_marginal_increase}
            \varphi_{m+1} \geq \frac{c}{k_m} d_m+ \frac{1-c}{k_m} d^*_m.
        \end{align}
        
        Applying this to the right hand side of $(\ref{eq:induction_claim_submod_and_curvature_analysis})$ for $m+1$ yields 
        \vspace{-0.5em}
        \begin{eqnarray*} \label{eq:induction_submod_claim_1}
            \left(1 -  \frac{c}{k} \right)^{k-m-1} c\, d_{m+1} &-& \left(1-\left(1- \frac{c}{k}\right)^{k-m-1}\right) \frac{(1-c)k}{k_{m+1}}d^*_{m+1} \nonumber 
        \end{eqnarray*}
        \vspace{-1.4em}
        \begin{eqnarray*}
            &\overset{(\ref{eq:isolate_varphi_in_indunction})}{=}& \left(1-\frac{c}{k} \right)^{k-m-1} c \, d_{m} - \left(1-\left(1- \frac{c}{k}\right)^{k-m-1}\right) \frac{(1-c)k}{k_{m+1}}d^*_{m} \\[-0.02em]
            &&- \left[ \left(1-\frac{c}{k} \right)^{k-m-1} c- \bigg(1-\left(1- \frac{c}{k}\right)^{k-m-1}\bigg)  \frac{(1-c)k}{k_{m+1}} \chi_{m+1} \right] \varphi_{m+1}\\
            &\overset{(\ref{eq:non_negativity_of_factor_before_increment}),(\ref{eq:bound_on_marginal_increase})}{\leq}& \left(1-\frac{c}{k} \right)^{k-m-1} c \, d_{m} - \left(1-\left(1- \frac{c}{k}\right)^{k-m-1}\right) \frac{(1-c)k}{k_{m+1}}d^*_{m} \\[-0.02em]
            &&- \left[ \left(1-\frac{c}{k} \right)^{k-m-1} c\right] \left(\frac{c}{k_m} d_m+ \frac{1-c}{k_m} d^*_m \right) \\[-0.02em]
            &&+ \left[\bigg(1-\left(1- \frac{c}{k}\right)^{k-m-1}\bigg)  \frac{(1-c)k}{k_{m+1}} \chi_{m+1} \right] \left(\frac{c}{k_m} d_m+ \frac{1-c}{k_m} d^*_m \right) \nonumber\\
            &=& \left(1-\frac{c}{k} \right)^{k-m-1} \left(1- \frac{c}{k_m}\right) c\, d_{m} \nonumber\\[-0.02em]
            &&- \left(1-\left(1- \frac{c}{k}\right)^{k-m-1} \right) \frac{(1-c)k}{k_{m+1}} \left( 1-\chi_{m+1}\frac{1-c}{k_m}\right) d^*_{m} \nonumber\\[-0.02em]
            &&-\left(1-\frac{c}{k} \right)^{k-m-1} \frac{c(1-c)}{k_m} d^*_{m} \\[-0.02em]
            &&+ \left(1-\left(1- \frac{c}{k}\right)^{k-m-1} \right)\frac{(1-c)k}{k_{m+1}} \chi_{m+1} \frac{c}{k_m}d_m \nonumber\\
            &\overset{k \geq k_m}{\leq} & \left(1-\frac{c}{k} \right)^{k-m} c\, d_{m} \\[-0.02em]
            &&- \left(1-\left(1- \frac{c}{k}\right)^{k-m-1} \right) \frac{(1-c)k}{k_{m+1}} \left( \frac{k_m-\chi_{m+1}+c\chi_{m+1}}{k_m}\right)d^*_{m} \nonumber\\[-0.02em]
            &&-\left(1-\frac{c}{k} \right)^{k-m-1} \frac{c(1-c)}{k_m} d^*_{m} \\[-0.02em]
            &&+ \left(1-\left(1- \frac{c}{k}\right)^{k-m-1} \right)\frac{(1-c)k}{k_{m+1}} \chi_{m+1} \frac{c}{k_m} d_m \nonumber\\
            &\overset{(\ref{eq:reference_releation_of_k})}{=} & \left(1-\frac{c}{k} \right)^{k-m} c\, d_{m} \\[-0.02em]
            &&- \left(1-\left(1- \frac{c}{k}\right)^{k-m-1} \right) \frac{k}{k_{m+1}} \left( \frac{k_{m+1}+c\chi_{m+1}}{k_m}\right) (1-c)d^*_{m} \nonumber\\[-0.02em]
            &&-\left(1-\frac{c}{k} \right)^{k-m-1} \frac{c}{k_m} (1-c)d^*_{m} \\[-0.02em]
            &&+ \left(1-\left(1- \frac{c}{k}\right)^{k-m-1} \right)\frac{k}{k_{m+1}}\frac{c\chi_{m+1}}{k_{m}} (1-c)d_m \nonumber\\
            &=& \left(1-\frac{c}{k} \right)^{k-m}  c\, d_{m} \\[-0.02em]
            &&- \left[\left(1-\left(1- \frac{c}{k}\right)^{k-m-1} \right) \frac{k}{k_m} + \left(1-\frac{c}{k} \right)^{k-m-1} \frac{c}{k_m}\right] (1-c)d^*_{m} \nonumber\\[-0.02em]
            &&+ \left(1-\left(1- \frac{c}{k}\right)^{k-m-1} \right)\frac{k}{k_{m+1}}\frac{c\chi_{m+1}}{k_{m}} (1-c) (d_m-d^*_m) \nonumber\\
            &\overset{d_m \leq d^*_m}{\leq}& \left(1-\frac{c}{k} \right)^{k-m}  c\, d_{m} \\[-0.02em]
            &&- \left[\left(1-\left(1- \frac{c}{k}\right)^{k-m-1} \right) \frac{k}{k_m} + \left(1-\frac{c}{k} \right)^{k-m-1} \frac{c}{k_m}\right] (1-c)d^*_{m}  
        \end{eqnarray*}
        \begin{eqnarray*}
            &=& \left(1-\frac{c}{k} \right)^{k-m} c\, d_{m} - \left[1-\left(1- \frac{c}{k}\right)^{k-m-1} \left(1- \frac{c}{k}\right)\right] \frac{(1-c)k}{k_m} d^*_{m} \nonumber\\
            &=& \left(1-\frac{c}{k} \right)^{k-m} c\, d_{m} - \left(1-\left(1- \frac{c}{k}\right)^{k-m} \right) \frac{(1-c)k}{k_m} d^*_{m} \nonumber \\
            &\overset{(\ref{eq:induction_claim_submod_and_curvature_analysis})}{\leq}& 
            \left[ \left(1-\frac{c}{k} \right)^{k} - 1 + c \right] \gamma \hfunc(S),
        \end{eqnarray*}
        where we used the induction hypothesis for $m$.
        
        For $m=k$, $(\ref{eq:induction_claim_submod_and_curvature_analysis})$ yields
        \begin{align} \label{eq:result_of_induction_submod_claim}
            \left[ \left(1-\frac{c}{k} \right)^{k} - 1 + c \right] \gamma \hfunc(S)
            \geq c \left(\gamma \hfunc(S) -  \sum_{i=1}^{k} \varphi_{i} \right).
        \end{align}
        Rearranging we derive, using $\left(1+\frac{(-c)}{k}\right)^k \leq \e^{-c}$ for $k \geq 1$, $\abs{c}\leq k$, that
        \begin{align*}
            c\sum_{i=1}^{k} \varphi_{i} \overset{(\ref{eq:result_of_induction_submod_claim})}{\geq} \left(1-\left(1-\frac{c}{k}\right)^{k}\right) \gamma \hfunc(S) \geq \left(1-\frac{1}{\e^c}\right) \gamma \hfunc(S).\tag*{\raisebox{-2ex}{\qedhere}}
        \end{align*}
    \end{proof}
    
    With this at hand, we are ready to prove our main result.
    
    \mainupperbound*
    
    \begin{proof}
        Let $k \leq n$ and $S \subseteq E$ with $\ell \coloneqq\abs{\setH_{k}}=\abs{S}$.
        By \Cref{prop:symmetry_handling}, it suffices to show that
        \begin{align} \label{eq:rho_competitiveness_sub_mod_curv}
            \gamma^{-1}\left(1+ c\frac{\e^c}{\e^c-1} \right) \ffunc(\setH_{k}) \geq \ffunc(S).
        \end{align}
        Since $\varphi_i \geq 0$ by monotonicity of $\gfunc$ and $\hfunc$, and since $k \geq \ell$, we obtain
        \begin{align} \label{eq:double_greedy_bound_on_f_submod_ratio_and_curvature}
            \ffunc(\setH_{k})
            \overset{\text{Lem. }\ref{lem:gen_lower_bound_for_double_greedy}}{\geq} \sum_{i=1}^{k} \varphi_{i} 
            \geq \sum_{i=1}^{\ell} \varphi_{i}
            \overset{\text{Lem. }\ref{lem:double_greedy_bound_on_opt_submod_ratio_and_curvature}}{\geq} \frac{1}{c}\left(1-\frac{1}{\e^c}\right) \gamma \hfunc(S).
        \end{align}
        By monotonicity of $\gfunc$, we have
        \begin{align} \label{eq:bound_g_for_submod_ratio_curv}
            \gfunc(S^{\comp}) \leq \gfunc(E) 
            = \gfunc(\setH_{k} \cup \setH^{\comp}_{k}) 
            = \sum_{i:e^*_i \in \setH_{k}} \gfunc(e^*_i \mid \setH_{i-1} \cup \setH_{k}^{\comp}) 
            + \gfunc(\setH^{\comp}_{k})
        \end{align}
        Further, using $\setG_{i-1} \subseteq \setH_{i-1} \cup \setH_{k}^{\comp}$ and the definition of the algorithm, we derive
        \begin{align} \label{eq:bound_g_value_of_opt}
            \gfunc(S^{\comp})
            & \overset{(\ref{eq:bound_g_for_submod_ratio_curv})}{\leq} \sum_{i:e^*_i \in \setH_{k}} \gfunc(e^*_i \mid \setH_{i-1} \cup \setH_{k}^{\comp}) + \gfunc(\setH^{\comp}_{k}) \nonumber\\
            &\overset{(\ref{eq:def_gen_sub_mod_ratio})}{\leq} \gamma^{-1} \sum_{i:e^*_i \in \setH_{k}} \gfunc(e^*_i \mid \setG_{i-1}) + \gfunc(\setH^{\comp}_{k}) \nonumber\\
            &\overset{\text{alg.}}{\leq} \gamma^{-1} \sum_{i:e^*_i \in \setH_{k}} \hfunc(e^*_i \mid \setH_{i-1}) + \gfunc(\setH^{\comp}_{k}) \nonumber\\
            &=  \gamma^{-1} \hfunc(\setH_{k}) + \gfunc(\setH^{\comp}_{k}) \nonumber\\
            & \overset{\gamma \leq 1}{\leq} \gamma^{-1} \ffunc(\setH_{k}).
        \end{align}
        Together with (\ref{eq:double_greedy_bound_on_f_submod_ratio_and_curvature}), we obtain the desired approximation $(\ref{eq:rho_competitiveness_sub_mod_curv})$
        \begin{align*}
            \ffunc(S) = \hfunc(S) +\gfunc(S^{\comp}) 
            &\overset{(\ref{eq:double_greedy_bound_on_f_submod_ratio_and_curvature}),(\ref{eq:bound_g_value_of_opt})}{\leq} 
            \gamma^{-1}c\left(\frac{\e^c}{\e^c-1}\right)\ffunc(\setH_{k}) + \gamma^{-1}\ffunc(\setH_{k}) \\
            &\overset{\phantom{(6),(8)}}{=} \gamma^{-1}\left(1+ c\frac{\e^c}{\e^c-1} \right)\ffunc(\setH_{k}).
            \tag*{\raisebox{-2ex}{\qedhere}}
        \end{align*}
    \end{proof}

      We now show the lower bound on the competitive ratio for the case that $\gfunc$ and $\hfunc$ are submodular with curvature $c$. This proves the first part of \Cref{thm:lower}.
    
    \begin{restatable}{proposition}{theoremlowerboundcurvature}
        \label{thm:lower_bound_submodular}
        For~$\gfunc, \hfunc \colon 2^E \to \R_{\geq 0}$ monotone submodular functions with curvature $c \in [0,1]$ 
        no algorithm has a better competitive ratio than $1+\frac{c}{2-c}$.
    \end{restatable}
    \begin{proof}
        To show the lower bound of $1+\frac{c}{2-c}$, consider the following instance. Let $n\coloneqq\abs{E}$, $c\in[0,1]$ and fix an element $e^* \in E$. We define 
        \[\
        \hfunc(S) \coloneqq \gfunc(S) \coloneqq 
        \begin{cases}
            n-1 + (1-c) \abs{S\setminus\{e^*\}} & \text{if } e^* \in S,\\
            \abs{S} & \text{otherwise.}
        \end{cases}
        \]
        The marginal increase of the element $e^*$ is $\hfunc(e^* \mid S)=n-1-c\abs{S}$ if $e^* \notin S$ and~$\hfunc(e^* \mid S)=0$ otherwise, implying $\hfunc(e^* \mid S) \geq \hfunc(e^* \mid T)$ for $S \subseteq T\subseteq E$.
        The marginal increase of every other element~$e \neq e^*$ is $\hfunc(e \mid S)=1$ if $e^* \notin S$ and $\hfunc(e \mid S)=1-c$ otherwise, implying $\hfunc(e \mid S) \geq \hfunc(e \mid T)$ for all~$S \subseteq T\subseteq E$.
        Together, we have that $\hfunc$ and $\gfunc$ are submodular.
        Thus, their curvature is given by $(\ref{eq:def_curvature_submodular})$ as
        \[
        1-\min_{e \in E} \frac{\hfunc(e\mid E \setminus \{e\})}{\hfunc(\{e\})}= 1- \min \left\{ \frac{(1-c)(n-1)}{n-1}, \frac{1-c}{1} \right\} =c.
        \]
        We now show that no solution $\pi$ to this instance is better than $\smash{\frac{2(n-1)}{(2-c)(n-1)+c}}$-competitive.
        For $n\rightarrow \infty$ we then obtain the desired lower bound of $\frac{2}{2-c}=1+\frac{c}{2-c}$.
        We observe that the optimum value of $\ffunc(S)=\hfunc(S)+\gfunc(S^\comp)$ with $\abs{S}=1$ is given by
        \[
            \ffunc(\{e^*\})=\hfunc(\{e^*\})+\gfunc(E\setminus\{e^*\})=n-1+n-1=2(n-1),
        \]
        since, for $e\neq e^*$, we have
        \begin{align*}
            \ffunc(\{e\})=\hfunc(\{e\})+\gfunc(E\setminus\{e\})&=1+n-1+(1-c)(n-2) \\
            &=2n-cn-2+2c=(2-c)(n-1)+c.
        \end{align*}
        Similarly, for $\abs{S}=n-1$ the optimum value is 
        \[
            \ffunc(E \setminus \{e^*\})=\hfunc(E \setminus\{e^*\})+\gfunc(\{e^*\})=n-1+n-1=2(n-1),
        \]
        since, for $e\neq e^*$, we have
        \[
            \ffunc(E \setminus \{e\})=\hfunc(E \setminus \{e\})+\gfunc(\{e\})=n-1+(1-c)(n-2)+1
            =(2-c)(n-1)+c.
        \]
        Fixing $\pi$, we distinguish two cases. 
        If the first element in $\pi$ is $e \neq e^*$, we consider the resulting solution of size $1$ which has value $\ffunc(\{e\})=(2-c)(n-1)+c$.
        Therefore, in this case $\pi$ has competitive ratio at least $\frac{2(n-1)}{(2-c)(n-1)+c}$.
        If the last element in $\pi$ is $e \neq e^*$, we consider the solution of size $n-1$ which has value $\ffunc(E \setminus \{e\})=(2-c)(n-1)+c$.
        Thus, also in this case the competitive ratio is at least $\frac{2(n-1)}{(2-c)(n-1)+c}$.
    \end{proof}
    
    We also give a lower bound in terms of the generic submodularity ration; this proves the second part of \Cref{thm:lower}.
    
    \begin{restatable}{proposition}{theoremlowergenericsubmodularityratio}
        \label{thm:lower_bound_generic_submodularity_ratio}
        For monotone functions~$\gfunc, \hfunc \colon 2^E \to \R_{\geq 0}$ with generic submodularity ratio~$\gamma \in (0,1]$
        no algorithm has a better competitive ratio than $\frac{2}{3}+\frac{1}{3\gamma}$.
    \end{restatable}
    \begin{proof}
        To show the lower bound of $\frac{2}{3}+\frac{1}{3}\gamma^{-1}$, consider the following instance. Let $E=\{a,b,c\}$, $\gfunc, \hfunc \colon 2^E\rightarrow \R_{\geq0}$ and $\gfunc \coloneqq \hfunc$ with 
        \begin{align*}
            &\hfunc(\emptyset)=0, \qquad \hfunc(\{a\})=\hfunc(\{b\})=\hfunc(\{c\})=1, \qquad
            \hfunc(\{a,b\})=\hfunc(\{b,c\})=2,\\ 
            &\hfunc(\{a,c\})=1+\gamma^{-1} \qquad \text{and} \qquad \hfunc(\{a,b,c\})=2+\gamma^{-1}.
        \end{align*}
        By definition of $\hfunc$, we have $\hfunc(b \mid S)=1$ for all $S \subseteq \{a,c\}$. For $e, e' \in \{a,c\}$, $e \neq e'$ and~$S\subseteq E \setminus \{e\}$ we have $\hfunc(e \mid S)= \gamma^{-1}$ if $S=\{e'\}$ or $S=\{b,e'\}$ and $\hfunc(e \mid S)=1$ else.
        Thus, for $e \in \{a,c\}$ we have
        \[
            \min_{A,B \subseteq E: \hfunc(e \mid A \cup B)>0 } \frac{\hfunc(e \mid A)}{\hfunc(e \mid A \cup B)}= \min \left\{ \frac{1}{1}, \frac{\gamma^{-1}}{\gamma^{-1}}, \frac{1}{\gamma^{-1}} \right\}=\gamma
        \]
        and
        \[
            \min_{A,B \subseteq E: \hfunc(b \mid A \cup B)>0 } \frac{\hfunc(b \mid A)}{\hfunc(b \mid A \cup B)}= \min \left\{ \frac{1}{1} \right\}=1
        \]
        which yields that the generic submodularity ratio of $\gfunc$ and $\hfunc$ is $\gamma$.
        
        We show that no solution $\pi$ is better than $\frac{2}{3}+\frac{1}{3}\gamma^{-1}$-competitive.
        The optimum value of size~$1$ is $\ffunc(\{b\})=\hfunc(\{b\})+\gfunc(\{a,c\})=1+1+\gamma^{-1}=2+\gamma^{-1}$.
        The optimum value of size~$2$ is $\ffunc(\{a,c\})=\hfunc(\{a,c\})+\gfunc(\{b\})=1+\gamma^{-1}+1=2+\gamma^{-1}$.
        If the first element in $\pi$ is $e\neq b$, with $e' \in \{a,c\}, e'\neq e$ we compute the value of the solution of size $1$, which is $\ffunc(\{e\})=\hfunc(\{e\})+\gfunc(\{b,e'\})=1+2=3$.
        Likewise, if the last element in $\pi$ is~$e\neq b$, with $e' \in \{a,c\}, e'\neq e$, we compute the value of the solution of size $2$, which is~$\ffunc(\{b,e'\})=\hfunc(\{b,e'\})+\gfunc(\{e\})=2+1=3$.
        This yields in both cases that the competitive ratio is at least $\frac{2+\gamma^{-1}}{3}=\frac{2}{3}+\frac{1}{3}\gamma^{-1}$.
    \end{proof}

    \section{Gross Substitute Functions}\label{sub:gross_substitute}

    The gross substitute property originates from the economics literature \cite{kelso1982job}. 
    In this context, given a value function $f\colon 2^E \rightarrow \R_{\geq0}$ and a modular price function $p \colon 2^E \rightarrow \R$, one is interested in sets~$S$ maximizing the utility $f(S)-p(S)$.
    A function $f$ is gross substitute if, for every price function and for every subset~$S$ maximizing the utility before increasing some of those prices, all elements of~$S$ with non-increased price are still part of a utility-maximizing set after the prices are increased.
    \begin{definition}[Kelso and Crawford~\cite{kelso1982job}]
        A function $f \colon 2^E \rightarrow \R$ is \emph{gross substitute} if for all modular price functions $p, p' \colon 2^E \rightarrow \R$ with $p \leq p'$ and $S \in \argmax_{S\subseteq E}\{f(S) - p(S) \}$ there is~$S' \in \argmax_{S\subseteq E}\{f(S) - p'(S) \}$ with $S \cap \{e \mid p(\{e\})=p'(\{e\})\} \subseteq S'$.
    \end{definition}

    Paes Leme~\cite{leme2017gross} surveys several equivalent ways to characterize gross substitute functions. 
    We will work with the following characterization.
    \begin{lemma}[Paes~Leme~\cite{leme2017gross}] \label{lem:gross_substitute}
        A monotone set function $f \colon 2^E \rightarrow \R_{\geq0}$ is \emph{gross substitute} if and only if it is submodular and for all $A, B \subseteq E$ with $A \cap B=\emptyset$, $|B|\geq2$ and all $b \in B$ there is~$b' \in B \setminus \{b\} $ with
        \[
            f(b \mid A) + f(B\setminus \{b\} \mid A) \leq  f(b' \mid A) + f(B \setminus \{b'\} \mid A).
        \]
    \end{lemma}

    The greedy algorithm is optimal for the incremental maximization problem with gross substitute objective function.
    Together with \Cref{prop:2rho} we obtain the following result for the randomized algorithm discussed in \Cref{sec:random} that returns uniformly at random the greedy order with respect to~$h$ or the reverse greedy order with respect to~$g$.
    
    \begin{corollary}
        For $\gfunc, \hfunc \colon 2^E \to \R_{\geq 0}$ monotone and gross substitute the randomized algorithm is $2$-competitive.
    \end{corollary}

    However, we obtain the same bound for gross substitute functions with the deterministic double-greedy algorithm.
    To show this, we need the following observation, which even holds for submodular functions. Recall that $\ffunc(S) = \hfunc(S) + \gfunc(S^{\comp})$.
    
    \begin{restatable}{lemma}{lemmaincreasingdecreasing}
        \label{lem:incresing_decreasing_double_greedy}
        For $\gfunc, \hfunc \colon 2^E \to \R_{\geq 0}$ submodular and $\prec \ \in \{\leq,<\}$, the double-greedy algorithm ensures that
        \begin{align*}
            \gfunc(E) &= \ffunc(\setH_0) \leq \dots \leq \ffunc(\setH_n) = \ffunc(\setG_n^{\comp}) \geq \dots \geq \ffunc(\setG_0^{\comp}) = \hfunc(E).
        \end{align*}
    \end{restatable}
    \begin{proof}
        By definition of $\ffunc$, we have $\ffunc(\setH_0)=\ffunc(\emptyset)=\hfunc(\emptyset)+\gfunc(E)=\gfunc(E)$ and, likewise, we have $\ffunc(\setG_0^{\comp})=\hfunc(E)$.

        To show $\ffunc(\setH_{i-1}) \leq \ffunc(\setH_{i})$ for $1 \leq i \leq n$, 
        it suffices to consider an iteration of the algorithm in which we obtain $\setH_i$ by adding $e^*_i$ to $H_{i-1}$,
        since in every other iteration $\setH_{i-1}=\setH_{i}$ holds, implying trivially $\ffunc(\setH_{i-1}) = \ffunc(\setH_{i})$.
        Submodularity of $\gfunc$ and $\setG_{i-1} \subseteq \setH_i^{\comp}$ imply $\gfunc(e^*_i \mid \setH_{i}^{\comp}) \leq \gfunc(e^*_i \mid \setG_{i-1})$.
        Since $e^*_i$ is added to $\setH_{i-1}$, we have $\gfunc(e^*_i \mid \setG_{i-1}) \prec' \hfunc(e^*_i\mid \setH_{i-1})$ for $\prec' \ \in \{\leq,<\}$ with $\prec' \neq \prec$.
        Together, we obtain 
        \begin{align*}
            \ffunc(e^*_i \mid \setH_{i-1}) &= \ffunc(\setH_i) - \ffunc(\setH_{i-1}) \\
            &= \hfunc(H_i) - \hfunc(\setH_{i-1}) + \gfunc(\setH_i^{\comp}) - \gfunc(\setH_{i-1}^{\comp})\\
            &=\hfunc(e^*_i \mid \setH_{i-1})-\gfunc(e^*_i \mid \setH_i^{\comp})\\
            &\geq \gfunc(e^*_i \mid \setG_{i-1})-\gfunc(e^*_i \mid \setH_i^{\comp})\\
            &\geq 0.
        \end{align*}
        
        The inequalities $\ffunc(\setG_{i-1}^{\comp}) \leq \ffunc(\setG_i^{\comp})$ for $1 \leq i \leq n$ follow by symmetry, as the algorithm for~$\gfunc'=\hfunc$, $\hfunc'=\gfunc$ and $\prec'$ computes the reversed ordering.
    \end{proof}
    
    With this, we are able to prove our second main result.

    \thmuppergs*
    \begin{proof}
        Fix $k\leq n$ and let $S\subseteq E$ with $|S|\leq|\setH_k|$ elements.
        By \Cref{prop:symmetry_handling} it suffices to show that $2\ffunc(H_k)\geq \ffunc(S)= \hfunc(S)+\gfunc(S^{\comp})$. 
        By \Cref{lem:incresing_decreasing_double_greedy} and monotonicity of $\gfunc$ we have that $\ffunc(\setH_k) \geq \gfunc(E) \geq \gfunc(S^{\comp})$.
        Therefore, it is left to show that $\ffunc(H_k) \geq \hfunc(S)$.
        By monotonicity of~$\hfunc$, we have for all $i\leq n$ that
        \begin{align} \label{eq:bound_on_h_of_opt_with_increments_on_h}
            \hfunc(S) &\leq \hfunc(S \cup \setH_{i}) \nonumber\\
            &= \hfunc(S \mid \setH_i) + \hfunc(\setH_i) \nonumber\\
            &\leq \hfunc(S\mid \setH_i) + \sum_{j:e^*_j \in \setH_i} \hfunc(e^*_j \mid \setH_{j-1}) + \sum_{j:e^*_j \in G_i} \hfunc(e^*_j \mid \setH_{j-1}) \nonumber\\[-0.5em]
            &= \hfunc(S \mid \setH_{i})+ \sum^{i}_{j=1} \hfunc(e^*_j \mid \setH_{j-1}).
        \end{align}
        Let $i(S) \in \N$ smallest such that $\setH_{i(S)} \nsubseteq S$ or $\setG_{i(S)} \cap S \neq \emptyset$, 
        i.e., iteration $i(S)$ is the first iteration in which the choice of the algorithm is no longer consistent with first including the elements of~$S$. 
        In the following, we show that we can modify $S$ iteratively, either removing~$e^*_{i(S)}$ from~$S$ if $e^*_{i(S)} \in \setG_{i(S)}$ or exchanging $e^*_{i(S)}$ with an element in~$S$ if $e^*_{i(S)}\in \setH_{i(S)}$ without increasing the bound on $\hfunc(S)$ derived in $(\ref{eq:bound_on_h_of_opt_with_increments_on_h})$ for $i=i(S)-1$, until $S=\setH_{\ell}$ for some~$\ell \leq k$. 
        \begin{claim}
            Either $S=\setH_{i(S)-1}$ or there is $S'$ with 
            $i(S)<i(S')$, $\abs{S'} \leq \abs{\setH_k}$, and 
            \begin{align} \label{eq:claim_gross_analysis}
                \hfunc(S \mid \setH_{i(S)-1})+ \sum^{i(S)-1}_{j=1} \hfunc(e^*_j \mid \setH_{j-1}) \leq \hfunc(S' \mid \setH_{i(S')-1})+ \sum^{i(S')-1}_{j=1} \hfunc(e^*_j \mid \setH_{j-1}).
            \end{align}
        \end{claim}
        We can use this claim to prove the theorem as follows.
        Since we have $i(S) < i(S') \leq n$, after applying the claim iteratively with the resulting~$S'$ as next~$S$ finitely many times, we obtain $S'=\setH_{\ell}$ for $\ell \coloneqq i(S')-1$. 
        Since $\abs{\setH_{\ell}} = \abs{S'} \leq \abs{\setH_k}$, we have $\ell \leq k$. 
        Chaining the inequalities~$(\ref{eq:claim_gross_analysis})$ for $S$ and $S'$ iteratively, we obtain that 
        \begin{align} \label{eq:bound_on_h_of_opt_using_claim_iteratively}
            \hfunc(S) &\overset{(\ref{eq:bound_on_h_of_opt_with_increments_on_h})}{\leq} \hfunc(S \mid \setH_{i(S)-1})+ \sum^{i(S)-1}_{j=1} \hfunc(e^*_j \mid \setH_{j-1}) \overset{(\ref{eq:claim_gross_analysis})}{\leq} \hfunc(\setH_{\ell} \mid \setH_{\ell})+ \sum^{\ell}_{j=1} \hfunc(e^*_j \mid \setH_{j-1}).
        \end{align}
        Application of \Cref{lem:gen_lower_bound_for_double_greedy} concludes the proof of the theorem 
        \begin{align*}
            \hfunc(S) &\overset{(\ref{eq:bound_on_h_of_opt_using_claim_iteratively})}{\leq} \sum^{\ell}_{j=1} \hfunc(e^*_j \mid \setH_{j-1}) \leq \sum^k_{j=1} \max \{\hfunc(e^*_j \mid \setH_{j-1}), \gfunc(e^*_j \mid \setG_{j-1}) \} \overset{\text{Lem.} \ref{lem:gen_lower_bound_for_double_greedy}}{\leq} \ffunc(\setH_k).
        \end{align*}
        It remains to prove the claim.
        First, we observe that, for~$i \leq i'$ and $S \subseteq E$, we have
        \begin{align} \label{eq:compute_bound_of_claim_increasing_in_i}
            \hfunc(S \mid \setH_{i}) 
            &= \hfunc(S \cup \setH_{i}) - (\hfunc(\setH_{i'}) - \hfunc(\setH_{i'} \mid \setH_{i}) )  \nonumber \\
            &\leq \hfunc(S \cup \setH_{i'}) - \hfunc(\setH_{i'})+ \sum_{i< j\leq i':e^*_j \in H_{i'}} \hfunc(e^*_j \mid \setH_{j-1})  \nonumber \\[-1em]
            &\leq 
            \hfunc(S \mid \setH_{i'})+ \sum^{i'}_{j=i+1} \hfunc(e^*_j \mid \setH_{j-1}).
        \end{align}
        Thus, we obtain for $i\leq i'$ and $S \subseteq E$ that
        \begin{align} \label{eq:bound_of_claim_increasing_in_i}
            \hfunc(S \mid \setH_{i}) + \sum^{i}_{j=1} \hfunc(e^*_j \mid \setH_{j-1})
            &\overset{(\ref{eq:compute_bound_of_claim_increasing_in_i})}{\leq} 
            \hfunc(S \mid \setH_{i'})+ \sum^{i'}_{j=1} \hfunc(e^*_j \mid \setH_{j-1}).
        \end{align}
        
        Now, let $i\coloneqq i(S)$ and $S\neq \setH_{i-1}$ which implies $\setH_{i-1} \subsetneq S \subseteq \setG_{i-1}^{\comp}$. We consider two cases.
        First, assume that $e^*_i \in H_i$.  
        Then, by $\setG_i=\setG_{i-1}$ and definition of $i(S)$, $e^*_i$ is the only element in $\setH_i \setminus S$. 
        Together, $\setH_{i-1}\subsetneq S$ and $e^*_i \notin \setH_{i-1}$ imply that $B \coloneqq (S \setminus \setH_{i-1}) \cup \{e^*_i\}$ contains at least $2$ elements. 
        Thus, we can apply \Cref{lem:gross_substitute} with $b \coloneqq e^*_i$ and $A\coloneqq \setH_{i-1}$ yielding the existence of $b' \in S \setminus \setH_{i-1}$ with 
        \begin{multline} \label{eq:gross_substitute_property_in_analysis}
            \hfunc(e^*_i \mid \setH_{i-1}) + \hfunc(S \setminus \setH_{i-1} \mid \setH_{i-1})\\ \leq  \hfunc(b' \mid \setH_{i-1}) + \hfunc(S \setminus \setH_{i-1} \cup \{e^*_i\} \setminus \{b'\} \mid \setH_{i-1}).
        \end{multline}
        Since $S\subseteq \setG_{i-1}^{\comp}$ and $e^*_{i} \notin \setG_{i-1} \cup \setH_{i-1}$, we have $B \subseteq E \setminus (\setH_{i-1} \cap \setG_{i-1})$. By the choice of the algorithm in iteration $i$, $e^*_i$ maximizes $\hfunc(e^*_i \mid \setH_{i-1})$ in $B \subseteq E \setminus (\setH_{i-1} \cap \setG_{i-1})$, yielding
        \begin{align} \label{eq:choice_of_alg_in_gross_subs_analysis}
            \hfunc(e^*_i \mid \setH_{i-1}) \geq \hfunc(b' \mid \setH_{i-1}).
        \end{align}
        Together, we obtain
        \begin{align} \label{eq:greedy_element_in_gross_subs_analysis}
            \hfunc(S \mid \setH_{i-1}) = \hfunc(S \setminus \setH_{i-1} \mid \setH_{i-1}) \overset{(\ref{eq:gross_substitute_property_in_analysis}), (\ref{eq:choice_of_alg_in_gross_subs_analysis})}{\leq}  \hfunc(S \cup \{e^*_i\} \setminus \{b'\} \mid \setH_{i-1}).
        \end{align}
        Setting $S'\coloneqq S \cup \{e^*_i\} \setminus \{b'\}$, we have $\setH_{i} \subseteq S'$ by $\setH_{i-1}\subseteq S$ and $b' \notin \setH_{i}$. Because of $S\subseteq \setG_{i-1}^{\comp}$ and $e^*_{i} \notin G_{i-1}$, we have $S' \subseteq \setG_{i-1}^{\comp}= \setG_i^{\comp}$. Together, $\setH_{i} \subseteq S' \subseteq \setG_{i}^{\comp}$ implies $i(S')>i=i(S)$.
        In addition, we check $\abs{S'}=\abs{S}\leq \setH_k$. 
        The proof of the claim follows in this case by inferring
        \begin{align*}
            \hfunc(S \mid \setH_{i-1})+ \sum^{i-1}_{j=1} \hfunc(e^*_j \mid \setH_{j-1}) &\overset{(\ref{eq:greedy_element_in_gross_subs_analysis})}{\leq} 
            \hfunc(S' \mid \setH_{i-1})+ \sum^{i-1}_{j=1} \hfunc(e^*_j \mid \setH_{j-1})\\ 
            &\overset{(\ref{eq:bound_of_claim_increasing_in_i})}{\leq} \hfunc(S' \mid \setH_{i(S')-1})+ \sum^{i(S')-1}_{j=1} \hfunc(e^*_j \mid \setH_{j-1}).
        \end{align*}

        Now consider the case $e^*_i \in \setG_i$.  
        Then, $\setH_i=\setH_{i-1}\subseteq S$ and the definition of $i(S)$ imply $e^*_i$ is the only element in $\setG_i \cap S$. Consider $S'\coloneqq S \setminus \{e^*_i\} \subseteq \setG^{\comp}_i$. Then, $\setH_i=\setH_{i-1}\subseteq S'\subseteq \setG^{\comp}_i$ implies $i(S')>i=i(S)$. 
        In addition, we check $\abs{S'} < \abs{S} \leq \abs{H_k}$. 
        We have by $\setH_{i-1}\subseteq S'$ and submodularity of~$\hfunc$ that
        \begin{align} \label{eq:removing_element_from_G_by_submodularity}
            \hfunc(e^*_i \mid S') \leq \hfunc(e^*_i \mid \setH_{i-1}).
        \end{align} 
        Using $\setH_{i-1}\subseteq S'$ in the first step we can conclude the proof of the claim by inferring
        \begin{eqnarray*}
            \hfunc(S \mid \setH_{i-1})+ \sum^{i-1}_{j=1} \hfunc(e^*_j \mid \setH_{j-1}) &=& \hfunc(e^*_i \mid S') + \hfunc(S' \mid \setH_{i-1})+ \sum^{i-1}_{j=1} \hfunc(e^*_j \mid \setH_{j-1})\\[-1em]
            &\overset{(\ref{eq:removing_element_from_G_by_submodularity})}{\leq}& \hfunc(S' \mid \setH_{i-1})+ \sum^{i}_{j=1} \hfunc(e^*_j \mid \setH_{j-1})\\[-1em]
            &\overset{(\ref{eq:bound_of_claim_increasing_in_i})}{\leq}& \hfunc(S' \mid \setH_{i(S')-1})+ \sum^{i(S')-1}_{j=1} \hfunc(e^*_j \mid \setH_{j-1}).\qquad\quad\qedhere
        \end{eqnarray*}
        
    \end{proof}
    
    \begin{remark}
        The analysis of the double-greedy algorithm for gross substitute functions is tight. Even for modular functions, the algorithm is not better than $2$-competitive. 
        Consider the instance $E=\{a,b\}$, $\gfunc, \hfunc \colon 2^E \rightarrow \R_{\geq0}$ modular, given by $\gfunc(\{a\})=\hfunc(\{a\})=\hfunc(\{b\})=1$ and $\gfunc(\{b\})=0$, and $\prec\;=\;<$. Then, the double-greedy algorithm can choose $e^*_1=a$ computing~$\pi=(a,b)$. The solution of size $1$ has value $\ffunc(\{a\})=\hfunc(\{a\})+\gfunc(\{b\})=1$.
        But the optimum value for a solution of size $1$ is $\ffunc(\{b\})=\hfunc(\{b\})+\gfunc(\{a\})=2$. Thus, the algorithm is no better than $2$-competitive. 
    \end{remark}

   We finally state that no algorithm can be better than $5/4$-competitive which proves the third part of \Cref{thm:lower}.
    
    \begin{restatable}{proposition}{theoremlowerboundgrosssubstitutes}
        \label{thm:lower_bound_gross_substitutes}
        For $\gfunc, \hfunc \colon 2^E \to \R_{\geq 0}$ monotone and gross substitute functions no algorithm has a better competitive ratio than $5/4$.
    \end{restatable}
    \begin{proof}
        Consider the following set functions $\hfunc$ and $\gfunc$ on $E=\{a,b,c\}$.
        Let 
        \begin{align*}
            &\hfunc(\emptyset)=0, \qquad \hfunc(\{a\})=1, \qquad
            \hfunc(\{b\})=\hfunc(\{c\})=\hfunc(\{b,c\})=2 \qquad \text{and } \\
            &\hfunc(\{a,b\})=\hfunc(\{a,c\})=\hfunc(\{a,b,c\})=3.
        \end{align*}
        Similarly, let
        \begin{align*}
            &\gfunc(\emptyset)=0, \qquad \gfunc(\{b\})=1, \qquad 
            \gfunc(\{a\})=\gfunc(\{c\})=\gfunc(\{a,c\})=2 \qquad \text{and } \\
            &\gfunc(\{a,b\})=\gfunc(\{b,c\})=\gfunc(\{a,b,c\})=3.
        \end{align*}

        We show that both functions are gross substitute. 
        Since $\hfunc$ and $\gfunc$ differ only by swapping~$a$ and $b$, it suffices to consider $\hfunc$. 
        To show submodularity of $\hfunc$, we check whether the marginal values for each  $e \in E$ decrease, i.e., $\hfunc(e \mid S) \geq \hfunc(e \mid T)$ for $S \subseteq T \subseteq E$.
        For~$e=a$, this follows by~$\hfunc(a \mid S)=1$ for~$S\subseteq \{b,c\}$, and~$\hfunc(a \mid S)=0$ otherwise. 
        For~$e=b$, this follows by~$\hfunc(b \mid S)=2$ if~$S=\emptyset$, $\hfunc(b \mid S)=1$ if~$S=\{a\}$, and~$\hfunc(b \mid S)=0$ else for~$S\subseteq E$. 
        By~$\hfunc$ being symmetric in~$b$ and~$c$, we conclude submodularity of~$\hfunc$.
        
        Using that $\hfunc(\{a\}) +\hfunc(\{b,c\})=3 \leq 5 = \hfunc(\{b\}) +\hfunc(\{a,c\}) = \hfunc(\{c\}) +\hfunc(\{a,b\})$, we derive that for~$B=\{a,b,c\}$, $A=\emptyset$ and each~$e \in B$ there is $f \in B \setminus \{e\} $ with
        \[
            \hfunc(e \mid A) + \hfunc(B\setminus \{e\} \mid A) \leq  \hfunc(f \mid A) + \hfunc(B \setminus \{f\} \mid A).
        \]
        In addition, observe that the condition in \Cref{lem:gross_substitute} for $\abs{B}=2$ is always fulfilled with equality.
        We can therefore conclude that $\hfunc$ is gross substitute.

        Observe that the optimum value for $\ffunc(S)=\hfunc(S)+\gfunc(S^\comp)$ with $\abs{S}=1$ is~$\ffunc(\{c\})=2+3=5$. 
        With $\abs{S}=2$ the optimum value is $\ffunc(\{a,b\})=3+2=5$.
        We show that the competitive ratio~$\rho$ of all orderings~$\pi$ is at least $1.25$. 
        If $\pi$ starts with $a$ or $b$, the solution of size~$1$ has value $\ffunc(\{a\})=\hfunc(\{a\})+\gfunc(\{b,c\})=1+3=4$ or $\ffunc(\{b\})=\hfunc(\{b\})+\gfunc(\{a,c\})=2+2=4$, implying $\rho\geq\frac{5}{4}$ in both cases.
        If~$\pi$ ends with $a$ or $b$, the solution of size $2$ has value $\ffunc(\{b,c\})=\hfunc(\{b,c\})+\gfunc(\{a\})=2+2=4$ or $\ffunc(\{a,c\})=\hfunc(\{a,c\})+\gfunc(\{b\})=3+1=4$. Also, in those cases we have $\rho\geq\frac{5}{4}$.
        Since $\pi$ cannot both start and end with $c$, we obtain~$\rho\geq \frac{5}{4}$.
    \end{proof}

    \section{Discussion}
    
    In this work, we introduced the incremental--decremental maximization framework and established several competitive ratios for randomized and deterministic algorithm. 
    While our results provide new insights, some open questions remain.

    Concerning lower bounds for the double-greedy algorithm, for monotone and sub\-modular functions~$h$ and $g$, \Cref{thm:main1} establishes  $\smash{(1+\frac{\e}{\e-1})}$-competitiveness.
    The algorithm cannot perform better than that, as shown in \Cref{prop:lower_submod_bound_double_greedy} in the appendix.
    We do not know, whether the general upper bound of $\frac{1}{\gamma}(1 + c \frac{\e^c}{\e^c-1})$ for the setting where $h$ and $g$ have curvature $c$ and generic submodularity ratio $\gamma$ is tight and leave this question as an open problem.

    For the special case that $g(E)=h(E)$
    alongside submodularity and monotonicity, the double-greedy algorithm guarantees a $2$-competitive solution, which is optimal. 

    \begin{restatable}{proposition}
    {propositiongandhsame}
    \label{prop:gandhsame}
        For monotone and submodular set functions $g,h\colon 2^E\rightarrow \R_{\geq 0}$ with $g(E)=h(E)$, the double-greedy algorithm is exactly $2$-competitive.
    \end{restatable}
    \begin{proof}
        Due to the monotonicity and the assumption that $g(E)=h(E)$, it holds that $f(S)\leq 2h(E)$ for every subset $S\subseteq E$, implying that $\OPT_k\leq 2h(E)$ for every $k$.
        Because of \Cref{lem:incresing_decreasing_double_greedy}, the value of the double-greedy solution $S_k=\{e_1,\dots,e_k\}$ of size $k$ is at least~$h(E)$ if $g(E)=h(E)$.
        All in all, we conclude that $\OPT_k$ is at most twice as much as the value~$f(S_k)$ of the double-greedy solution of size $k$.
        According to \Cref{thm:lower_bound_submodular}, there is no algorithm which attains a better competitive ratio in this setting. 
    \end{proof}
    \noindent It would be interesting to see similar results which exploit dependencies between $h$ and $g$.

    Another open question is whether \Cref{thm:main1} can be strengthened in settings where $g$ and~$h$ have different curvatures $c_g\neq c_h$ and different generic submodularity ratios $\gamma_g\neq \gamma_h$. 
    By defining $c\coloneqq \max\{c_g,c_h\}$ and $\gamma \coloneqq \min\{\gamma_g, \gamma_h\}$  one can adapt the proofs leading to \Cref{thm:main1} obtaining $\frac{1}{\gamma}(1+c\frac{\e^c}{\e^c-1})$-competitiveness of the double-greedy algorithm.
    Note that the function space $\mathcal C$ in \Cref{prop:symmetry_handling} is then given by all monotone set functions on the ground set $E$  with curvature at most  $c$ and generic submodularity ratio at least $\gamma$.

    While the simple randomized algorithm stated in \Cref{cor:random} provides in some cases a better competitive analysis, it is outperformed by the double-greedy algorithm for submodular functions.
    Buchbinder et al.~\cite{Buchbinder2015tight} introduced a randomized version of the double-greedy algorithm improving the approximation factor for unconstrained submodular maximization from $3$ to $2$.
    Their improved analysis, however, does not directly lead to a better competitiveness for a randomized version of the double-greedy for incremental--decremental maximization as the worst case ratio is not necessarily attained at the position where $k$ equals the cardinality of a maximum set of $f$.
    We leave the design of randomized algorithms with better competitive ratio as an interesting problem for future research.

    Finally, another open problem  is to narrow or even close the gaps identified in \Cref{tab:overview} for deterministic algorithms, in particular, the gap between $2$ and $1+\frac{\e}{\e-1}$ for submodular $h$ and $g$, as well as the gap between $1.25$ and $2$ in the gross substitute setting.    
    
    \bibliographystyle{plainurl}
    \bibliography{literature}
    
    \appendix
    \section{Appendix}

    \subsection{Unboundedness of the Incremental Maximization Problem}
    \label{app:inc_max_unbounded}
    \begin{proposition} \label{prop:inc_max_unbounded}
        The competitive ratio of the incremental maximization problem with non-monotone submodular objective function $f\colon 2^E\rightarrow \mathbb R_{\geq 0}$ is unbounded.
    \end{proposition}

    \begin{proof}
        Let $n\coloneqq \abs{E}$.
        Fix $e^* \in E$ and consider the set function $f \colon 2^E \rightarrow \R_{\geq0}$ for $0<\varepsilon<1$ given by
    \[
        f(S) \coloneqq 
        \begin{cases}
            1 & e^* \in S,\\
            \abs{S}\cdot \varepsilon & \text{else.} 
        \end{cases}
    \]
    Consider an algorithm that adds element $e^*$ first. 
    Thus, the solution of the algorithm is $1$ for every $k$ while the optimum value is $k\varepsilon$ for $k \geq \frac{1}{\varepsilon}$.
    On the other hand, an algorithm that does not add $e^*$ first has value $\varepsilon$ for $k=1$ but the optimum solution is $1$. 
    For $\varepsilon\coloneqq\frac{1}{\sqrt{n}}$ considering steps $k=1$ and $k=n$ one can see that all algorithms have competitive ratio at least $ \sqrt{n}$ which is unbounded for $n \rightarrow \infty$.
    \end{proof}    

    \subsection{Tight lower bound for the double-greedy algorithm\texorpdfstring{\newline}{-} for submodular functions}
    \label{app:lower_bound_double_greedy}
    A well-known example for monotone submodular functions are coverage functions. 
    They count how many elements are covered by a set of subsets of those elements.
    \begin{definition}
        Let $E \subseteq 2^U$ be a set of subsets of a finite ground set $U$.
        Then, $f \colon 2^E \rightarrow \N$ given by
        $f(S) \coloneqq \left|\bigcup_{C \in S} C \right|$ for $S\subseteq E$ is a \emph{coverage function}.
    \end{definition}
    Submodularity of a coverage function $f$ can be seen by observing that the marginal increase~$f(C \mid S)$ of a set $C$ on some set of subsets $S\subseteq E$ is exactly the number of elements in~$C$ that are not already covered by subsets in $S$. When adding more subsets to $S$, the set of elements already covered only grows. Therefore, the marginal increase of $C$ on $S$ decreases, when more subsets are added to $S$.
    
    \begin{proposition}
        The competitive ratio of the double-greedy algorithm is exactly $(1+\frac{\e}{\e-1})$ for $h$ monotone and submodular and for $g$ modular and monotone.
    \end{proposition}
    \begin{proof}
        By Theorem \ref{thm:main1} the competitive ratio is at most $\smash{\big(1+\frac{\e}{\e-1}\big)}$. It is left to show that the algorithm is not $\rho$-competitive for any $\rho<\smash{\big(1+\frac{\e}{\e-1}\big)}$.
        
        To this end, consider the following instance. 
        Let $U\coloneqq\{1, \dots, k\}^k$ be the ground set of~$k$-tuples with entries $1, \dots, k$. 
        For $i \in \{1,\dots, k\}$ we define sets $A_i\coloneqq \{1, \dots, k\}^{k-1} \times \{i\}$ to partition the ground set into $k$ pairwise disjoint sets of equal size by fixing the last entry to $i$.
        We define sets $B_i\coloneqq \{1, \dots, k\}^{i-1}\times \{k\} \times \{1, \dots, k\}^{k-i}$ of the same size by fixing the $i$-th entry to~$k$ for $i \in \{1, \dots,k\}$. 
        Note that the sets $A_i$ and $B_i$ all have $k^{k-1}$ elements for $i \in \{1, \dots, k\}$.
        Let $h\colon 2^E \rightarrow \N$ be a coverage function on $U$ with $E\coloneqq \{A_1, \dots, A_k, B_1, \dots, B_k\}$.
        Then, for $S\subseteq E$ the value $h(S)\coloneqq \abs{\bigcup_{C \in S} C} $ is the number of elements covered by the sets in~$S$.
        Since~$h$ is a coverage function, it is monotone and submodular. 
        Let $g\colon 2^E \rightarrow \N$ be the modular function with $g(\{B_i\})=h(B_i \mid \{B_1, \dots, B_{i-1}\})$ and $g(\{A_i\})=0$ for $i \in\{1, \dots, k\}$. Note that, $g(\{B_1, \dots B_k\})=h(\{B_1, \dots, B_k\})$ by definition of $g$.

        The optimum solution of size $k$ for $f$ is $S^*_k\coloneqq\{A_1,\dots,A_k\}$ with complement $(S^*_k)^{\comp}\coloneqq\{B_1,\dots,B_{k}\}$. 
        Since $S^*_k$ covers the ground set completely and $(S^*_k)^{\comp}$ covers all elements apart from the $(k-1)^k$ elements with no entry equal to $k$,
        \[f(S^*_k)=h(S^*_k)+g((S^*_k)^{\comp})=k^k + h\left((S^*_k)^{\comp}\right)= k^k+k^k-(k-1)^k =2k^k-(k-1)^k.\]

        On the other hand $B_1, \dots, B_k$ is a possible solution of the double-greedy algorithm of size $k$. 
        By modifying the values of $h$ and $g$ by $\varepsilon$ one can force the algorithm to resolve the tie-breaking accordingly and reach the same lower bound for $\varepsilon \rightarrow 0$. 
        In the following, we choose how the ties are resolved by the algorithm for simplicity.
        Consider iteration~$\ell \leq k$ of the algorithm under the assumption that $H_{\ell-1}=\{B_1, \dots, B_{\ell-1}\}$ and  $G_{\ell-1}=\emptyset$ are already computed. 
        Then, for $i \in \{1, \dots, k\}$ and $j \in\{\ell, \dots,  k\}$ the sets $A_i$ and $B_j$ cover exactly $(k-1)^{\ell-1}k^{k-\ell}$ elements which are not already covered by $H_{\ell-1}$. 
        These are the elements $(x_1, \dots, x_k) \in \{1, \dots, k\}^k$ with $x_k=i$ for $A_i$ and $x_j=k$ for $B_j$, respectively, and $x_r\neq k$ for $r\in \{1, \dots, \ell-1\}$. 
        Therefore,~$B_\ell$ maximizes $h(B_\ell \mid H_{\ell-1})$ in $H_{\ell-1}^{\comp}$. 
        Since $g(\{B_i\})=h(B_{i}\mid \{B_1, \dots, B_{i-1}\})= (k-1)^{i-1}k^{k-i}$ is decreasing in $i\in\{\ell, \dots, k\}$, $B_\ell$ also maximizes $g(\{B_{\ell}\})$ in $H_{\ell-1}^{\comp}$. 
        Because of $h(B_\ell \mid H_{\ell-1})=g(\{B_{\ell}\})$ the algorithm can choose~$B_\ell$ as the next element setting $H_{\ell}\coloneqq H_{\ell-1} \cup \{B_{\ell}\}$ and $G_{\ell}\coloneqq G_{\ell-1}=\emptyset$. 
        Thus, $B_1, \dots, B_k$ is a possible solution of the double-greedy algorithm. 
        In this case, since $H_k$ covers all elements apart from the $(k-1)^k$ elements with no entry equal to $k$, we have with $g(H_k^{\comp})=0$ that
        \[f(H_k)=h(H_k)+g(H_k^{\comp})=k^k -(k-1)^k.\]
        Thus, we conclude
        \[ \lim_{k\rightarrow \infty} \frac{f(S^*_k)}{f(H_k)}
        =\lim_{k\rightarrow \infty} \frac{2-\left(\frac{k-1}{k}\right)^k}{1 -\left(\frac{k-1}{k}\right)^k}
        =\frac{2-\e^{-1}}{1-{\e^{-1}}} = \frac{2\e-1}{\e-1} = 1+\frac{\e}{\e-1}, \]
        which shows the competitive ratio of the algorithm is at least $1+\frac{\e}{\e-1}$ and therefore by Theorem \ref{thm:main1} exactly $1+\frac{\e}{\e-1}$.
    \end{proof}

    Since modular functions are submodular, we obtain the following Corollary.
    
    \begin{corollary}
        \label{prop:lower_submod_bound_double_greedy}
        The competitive ratio of the double-greedy algorithm is exactly $(1+\frac{\e}{\e-1})$ for~$h$ and $g$ monotone and submodular.
    \end{corollary}
\end{document}